\documentclass[12pt]{article}
\usepackage[margin=1in]{geometry}
\usepackage{setspace}
\usepackage{algorithm}
\usepackage{algpseudocode}
\usepackage{tikz}
\usepackage{enumitem}
\usepackage{caption} 
\usepackage{endnotes}

\usepackage{amssymb}
\usepackage{amsmath}
\usepackage{amsthm}
\usepackage{hyperref}
\usepackage{cleveref}
\hypersetup{colorlinks=true,citecolor=black,linkcolor=black,urlcolor=black,breaklinks=true}
\usepackage{subcaption}
\usepackage{bm}

\usepackage{natbib}
\bibpunct[, ]{(}{)}{,}{a}{}{,}%
\usepackage{etoolbox}

\providecommand{\keywords}[1]{\small	\textbf{\textit{Keywords---}} #1}

\newtheorem{lemma}{Lemma}
\newtheorem{proposition}{Proposition}

\newtheorem{theorem}{Theorem}
\newtheorem*{theorem*}{Theorem}
\newtheorem*{proposition*}{Proposition}

\theoremstyle{definition}

\DeclareMathOperator*{\argmax}{arg\,max}

\title{Package Bids in Combinatorial Electricity Auctions: Selection, Welfare Losses, and Alternatives}

\author{Thomas H\"ubner\thanks{Power Systems Laboratory, ETH Z\"urich, Switzerland, Email: thuebner@ethz.ch}, \: Gabriela Hug\thanks{Power Systems Laboratory, ETH Z\"urich, Switzerland, Email: ghug@ethz.ch}}

\date{}

\begin{document}

\maketitle

\vspace{-0.5cm}

\begin{abstract}
A key challenge in combinatorial auctions is designing bid formats that accurately capture agents' preferences while remaining computationally feasible. 
This is especially true for electricity auctions, where complex preferences complicate straightforward solutions.
In this context, we examine the XOR package bid, the default choice in combinatorial auctions and adopted in European day-ahead and intraday auctions under the name ``exclusive group of block bids''.
Unlike parametric bid formats often employed in US power auctions, XOR package bids are technology-agnostic, making them particularly suitable for emerging demand-side participants.
However, the challenge with package bids is that auctioneers must limit their number to maintain computational feasibility.
As a result, agents are constrained in expressing their preferences, potentially lowering their surplus and reducing overall welfare.
To address this issue, we propose decision support algorithms that optimize package bid selection, evaluate welfare losses resulting from bid limits, and explore alternative bid formats. 
In our analysis, we leverage the fact that electricity prices are often fairly predictable and, at least in European auctions, tend to approximate equilibrium prices reasonably well.
Our findings offer actionable insights for both auctioneers and bidders.
\end{abstract}


\section{Introduction}

Electricity trading is constrained by complex technical requirements in transmission, production, and consumption,  which underscores the need for \textit{auctions} to facilitate efficient resource allocation and price discovery~\citep{milgrom2017discovering}. Most liberalized power markets feature at least a day-ahead auction, while some, such as the coupled European market, also offer multiple intraday auctions~\citep{graf2024frequent}.
In these auctions, the traded commodities are units of electric power~(MW) for specified future periods. For example, a day-ahead auction might facilitate buying or selling electricity for the 8-9 am period tomorrow, while an intraday auction could allow trading for the 3:15-3:30 pm slot later within the day.
This article addresses the following central question in \textit{electricity auction design}: How can participants best express their preferences to buy or sell future electricity? 

Preferences are expressed through bids, which can take various forms, but must conform to the \textit{bid formats} defined by the auctioneer. Auctions that enable participants to bid on combinations of goods, rather than individual ones, are known as combinatorial auctions~\citep{pekevc2003combinatorial}. Coupled European day-ahead and intraday auctions are examples of such combinatorial auctions, employing the default bid format for these types of auction, the XOR package bid, under the name ``exclusive group of block bids''~\citep{herrero2020evolving}.\endnote{Initially, block bids were limited to constant power profiles across periods, but this changed with the introduction of profile block bids~\citep{herrero2020evolving}. Although terms for package and XOR bids vary between power exchanges, the concept remains consistent. Integrated into the EUPHEMIA market coupling algorithm, XOR bids are available for use by exchanges in all member states~\citep{euphemia}. However, they are most commonly utilized by EpexSpot and Nordpool, which operate in 16 European countries. In 2023, on average, 158 exclusive groups were submitted daily in the European coupled day-ahead auction~\citep{cacm_report_2023}.} This format allows agents to bid on multiple power profiles that they are willing to buy or sell at a given price $p$, while ensuring that the auctioneer can accept at most one of these profiles~(\Cref{fig: xor bid illustration}).

\begin{figure}[t]
\caption{Illustration of an exclusive group of three block bids (XOR package bids).\protect\endnote{The power exchange EpexSpot restricts block bids from including both sell and buy quantities for different time periods, which prevents arbitrageurs from utilizing XOR bids as demonstrated in~\Cref{subfig: illustration arbitrageur}. This limitation is based on historical practices rather than any technical necessity. While so-called loop blocks enable the linkage of a buy and a sell profile~\citep{karasavvidis2023optimal}, they cannot be grouped within an exclusive group. A better approach would be simply allowing a profile to contain both buy and sell quantities, eliminating this unnecessary restriction.}}
\centering
    \begin{subfigure}[c]{0.32\textwidth} 
        \centering
        \caption{Demand side.}
        \label{subfig: illustration demand side}
        \includegraphics[height=3.1cm]{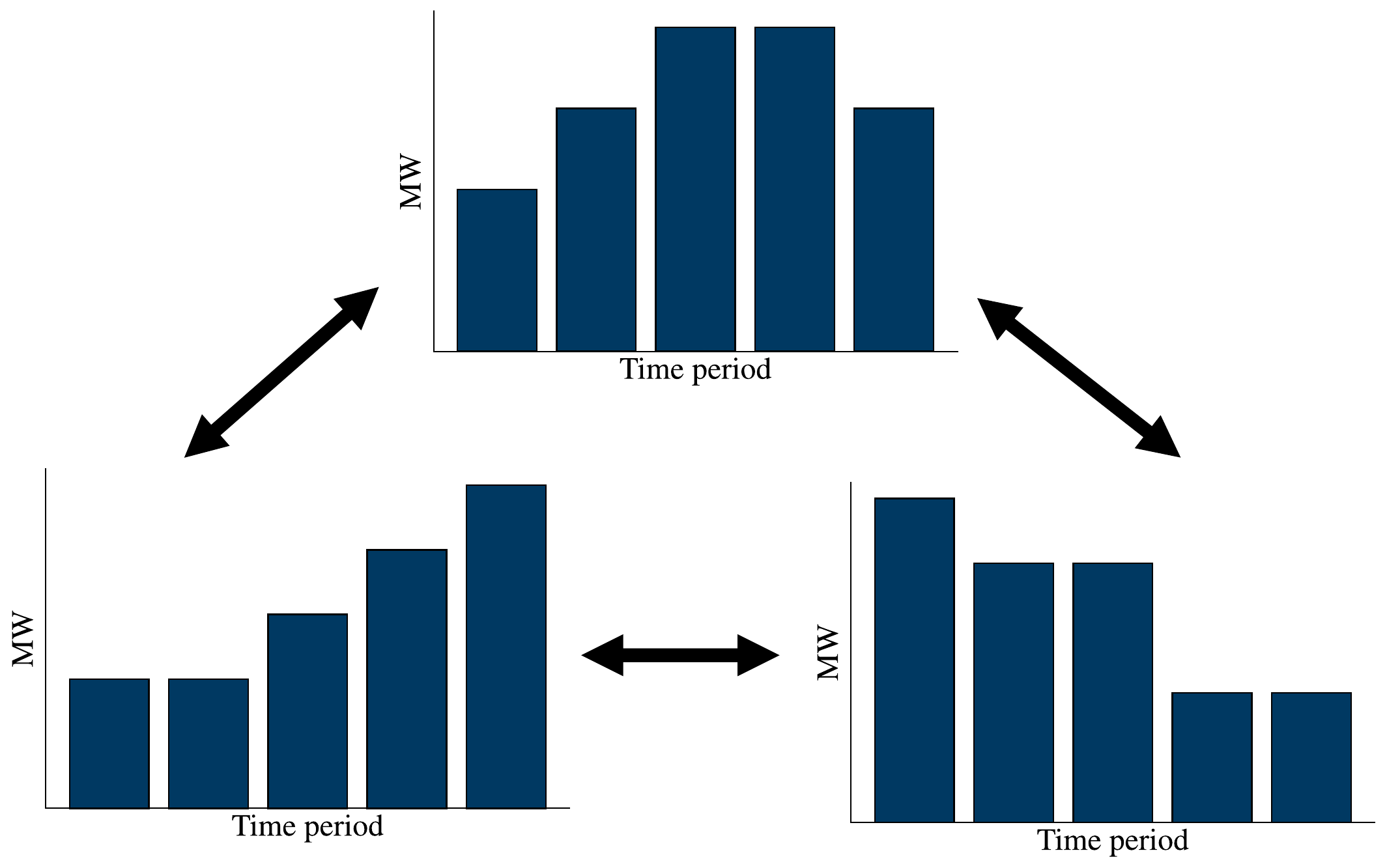}
    \end{subfigure}
    \begin{subfigure}[c]{0.32\textwidth} 
        \centering
        \caption{Arbitrageur.} 
        \label{subfig: illustration arbitrageur}
        \includegraphics[height=3.1cm]{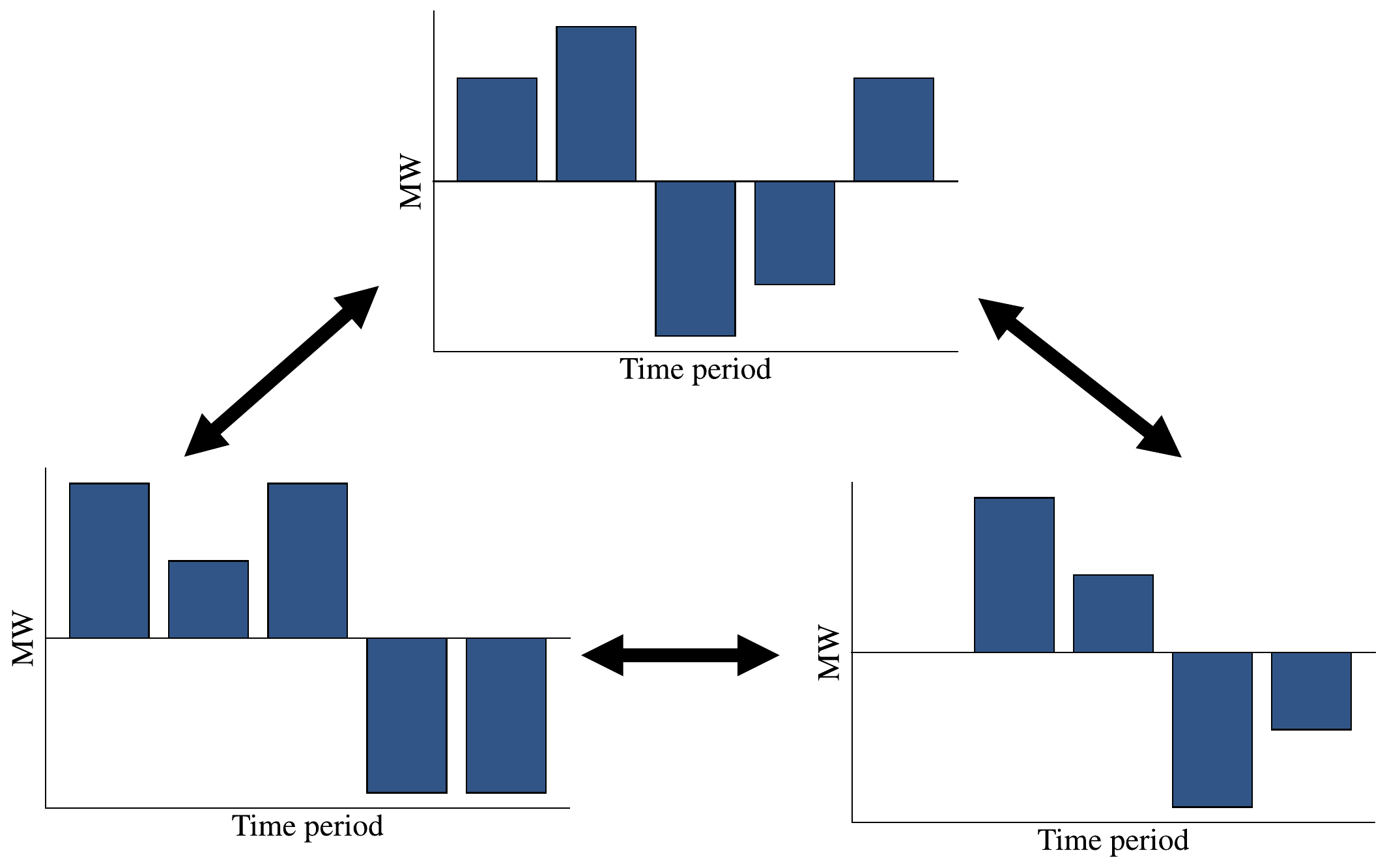}
    \end{subfigure}
    \begin{subfigure}[c]{0.32\textwidth} 
        \centering
        \caption{Supply side.}
        \label{subfig: illustration supply side}
        \includegraphics[height=3.1cm]{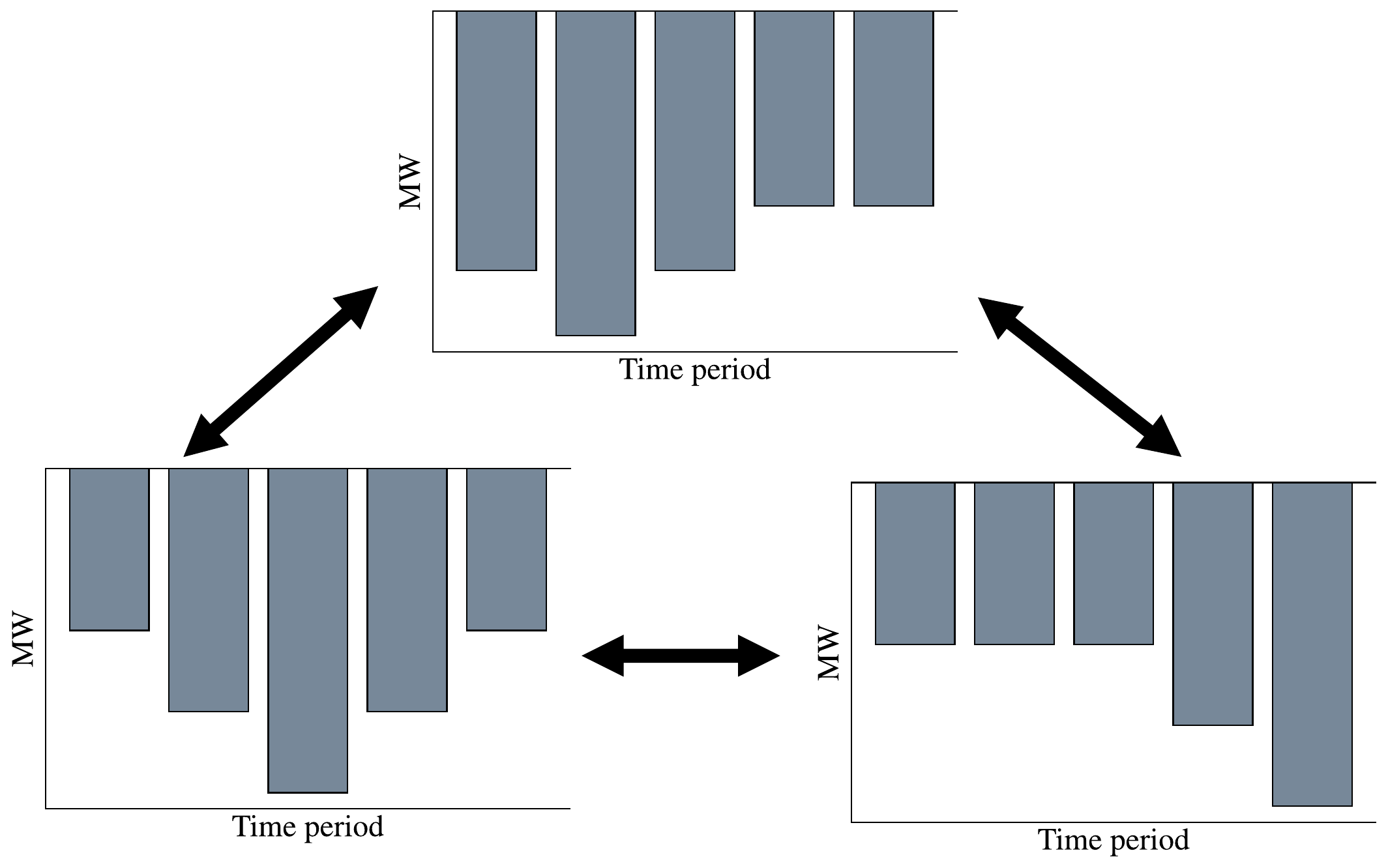}
    \end{subfigure}

\label{fig: xor bid illustration}
\end{figure}

The rationale for using XOR package bids in electricity auctions stems from the intertemporal constraints that many participants face. For example, a battery storage system needs to charge before it can discharge, a shiftable load may prefer to consume in period x or y, but not both, and a thermal generator may need to operate over multiple periods continuously. Participants can ensure that these constraints are met by bidding on entire power profiles instead of separate quantities for each period. Declaring these package bids as XOR bids - grouping multiple profiles into an exclusive set - allows participants to avoid selling more power than they can produce or buying more than they can consume, as only one profile is accepted at most.
XOR bids are capable of representing any type of preference~\citep{nisan2006bidding}, making them a versatile tool for agents on both the demand and supply sides, as well as for arbitrageurs (\Cref{fig: xor bid illustration}).

In combinatorial auctions, the auctioneer typically uses an optimization model to determine which bids to accept and the corresponding payments to agents. To ensure the computational tractability of this model, the auctioneer often limits the number of XOR package bids that can be submitted~\citep{goetzendorff2015compact}. This restriction creates a \textit{missing bids problem}, where agents cannot fully express their preferences - some packages they are interested in cannot be bid on. As a result, agents can miss potential profits, leading to a loss in society's welfare~\citep{bichler2023taming}.
This issue is also evident in the coupled European day-ahead electricity auction, which imposes a limit of 24 block bids per exclusive group~\citep{epexspot}. In this work, we explore the effects of this limit and aim to answer the following questions:
\begin{itemize}
\item How can participants \textbf{select package bids} that maximize their profit?\endnote{For simplicity, we use the term \textit{profit} to encompass all types of participants. For demand-side agents, this typically refers to minimizing procurement costs.}
\item What are the \textbf{welfare implications} of a limit on XOR bids?
\item Could \textbf{alternative bid formats} help solve the missing bid problem?
\end{itemize}

To address these questions, we must first examine the mechanisms employed by auctioneers to determine accepted bids and payments. Electricity auctions typically use a \textit{Walrasian mechanism}~\citep{milgrom_watt_2022}.\endnote{The widespread adoption of Walrasian mechanisms in electricity auctions is primarily driven by three factors: (i) the fairness and simplicity of uniform pricing, which ensures a single price per good for all agents~\citep{bichler2018matter}, (ii) its ability to serve as a clear scarcity signal, in contrast to nonlinear price signals that provide limited transparency~\citep{ahunbay2024pricing}, and (iii) its facilitation of financial arbitrage across forward markets, from years ahead to real-time, through derivatives~\citep{jha2023can}.} Although specific implementations may vary, these mechanisms aim to find uniform prices that clear the market. That is,
\begin{enumerate}[label=(\roman*)]
    \item Uniform prices (€/MW) are determined for electricity in each time period. Each agent pays the product of the accepted bid quantities and the uniform prices for each period.
    \item Only the most profitable XOR package bid is accepted; all are rejected if none are profitable.\endnote{For different bid formats, this rule may be stated differently. However, the main principle is that, at a Walrasian equilibrium, bids are accepted or rejected so that no agent would prefer a different outcome, given the equilibrium prices~\citep{mas1995microeconomic}.}
\end{enumerate}

If the auctioneer consistently identifies such a \textit{Walrasian equilibrium} and no agent has market power, which means that no agent can influence uniform prices, this mechanism is \textit{strategy-proof}~\citep{azevedo2019strategy}. In a strategy-proof mechanism, truthful bidding becomes dominant, as agents need not strategize based on others' bids. Furthermore, the allocation achieved under such conditions maximizes welfare in the auction~\citep{mas1995microeconomic}.

However, the two ideal conditions - (i) a perfectly competitive market and (ii) the existence of a Walrasian equilibrium - cannot be guaranteed. In the first case, this discrepancy requires the use of market power mitigation methods~\citep{adelowo2024redesigning}. For the second case, the Walrasian mechanisms employed can only approximate a Walrasian equilibrium, where not every agent is allocated their most profitable bid. While the European algorithm EUPHEMIA does not provide side payments to agents with such ``paradoxical'' outcomes~\citep{euphemia}, various forms of side payments are commonly incorporated into the mechanisms used by auctioneers in the United States~\citep{stevens2024some}.

In addressing the above three questions on XOR bids, we assume those two ideals are true: (i) Walrasian equilibria exist, and (ii) no agent has market power. 
While these assumptions are admittedly strong, discussions with industry stakeholders involved in European auctions - along with supporting evidence from~\cite{hubner2025approximate} and~\cite{graf2013measuring} - suggest that such assumptions are common in practice and approximate real-world conditions reasonably well. 
Notably, the study by \cite{karasavvidis2024optimal} - to the best of our knowledge the only one to address XOR bidding in electricity auctions - also adopts these assumptions.\endnote{Studies on XOR bid selection in other domains include \cite{scheffel2012impact}, which examines bounded rationality in spectrum auctions, or \cite{hammami2021exact}, which addresses the ``bid construction problem'' in transportation auctions. However, we did not find any direct transferable approach to electricity auctions. A key difference might be the high-quality price information available in daily electricity auctions. Furthermore, distinct auction mechanisms, such as pay-as-bid or multi-round formats, require specialized decision support strategies~\citep{pekevc2003combinatorial, adomavicius2022bidder}. Combinatorial auctions are typically domain specific and often require customized solutions~\citep{bichler2023taming}.
Most studies on bidding in electricity auctions focus on hourly bids~\citep[see the review by][]{aasgaard2019hydropower}, while block bids have received comparatively less attention~\citep[e.g.,][]{fleten2007stochastic,karasavvidis2021optimal}. To the best of our knowledge, \citet{karasavvidis2024optimal} is the only study that explicitly examines XOR bids in electricity auctions.} 
In their work, \cite{karasavvidis2024optimal} formulate the XOR bid selection problem as a \textit{stochastic bilevel optimization} model, where uniform prices are given as exogenous scenarios, and the lower level optimization reflects the bid acceptance rule under the assumed Walrasian equilibrium.\endnote{\cite{karasavvidis2024optimal} do not mention the possibility that this rule may not always hold, but instead implicitly assume equilibrium existence. Industry representatives note that paradoxically rejected or accepted bids are rarely an issue in the European day-ahead market. Indeed, in 2023, an average of 4,401 block bids were submitted per day, while only 7.5 were paradoxically rejected~\citep{cacm_report_2023}. A more comprehensive study on this phenomenon can be found in~\cite{hubner2025approximate}.}

Given these two idealized assumptions, it seems logical to submit as many XOR bids as possible to hedge against uncertain prices. However, \cite{karasavvidis2024optimal} use their method to determine only ten package bids without mentioning the upper limit of 24. Similarly, an analysis of XOR bids in the German \& Luxembourgian bidding zone in 2023 shows that many participants do not exploit this limit either~(\Cref{fig: number of bids}).\endnote{Perhaps unsurprisingly, a consultant from the energy industry shared that when they highlight the benefits of submitting the maximum number of bids, traders quickly see its value and wonder why they hadn’t thought of it sooner. The auction literature contains numerous examples where participants do not bid optimally due to a misunderstanding of auction mechanisms~\citep{milgrom2017discovering}. The concern that combinatorial auctions are too complex for agents to participate efficiently is as old as the idea of combinatorial auctions itself~\citep{rassenti1982combinatorial}.}

\begin{figure}[t]
\caption{Usage of exclusive groups (XOR bids) in 2023 in the Germany \& Luxembourg bidding zone.\protect\endnote{The historical bid data from the day-ahead auction was purchased from the power exchange EpexSpot.}}
\centering
    \begin{subfigure}[c]{0.45\textwidth} 
        \centering
        \caption{Demand bids.} 
        \label{subfig: demand number of bids}
        \includegraphics[height=5cm]{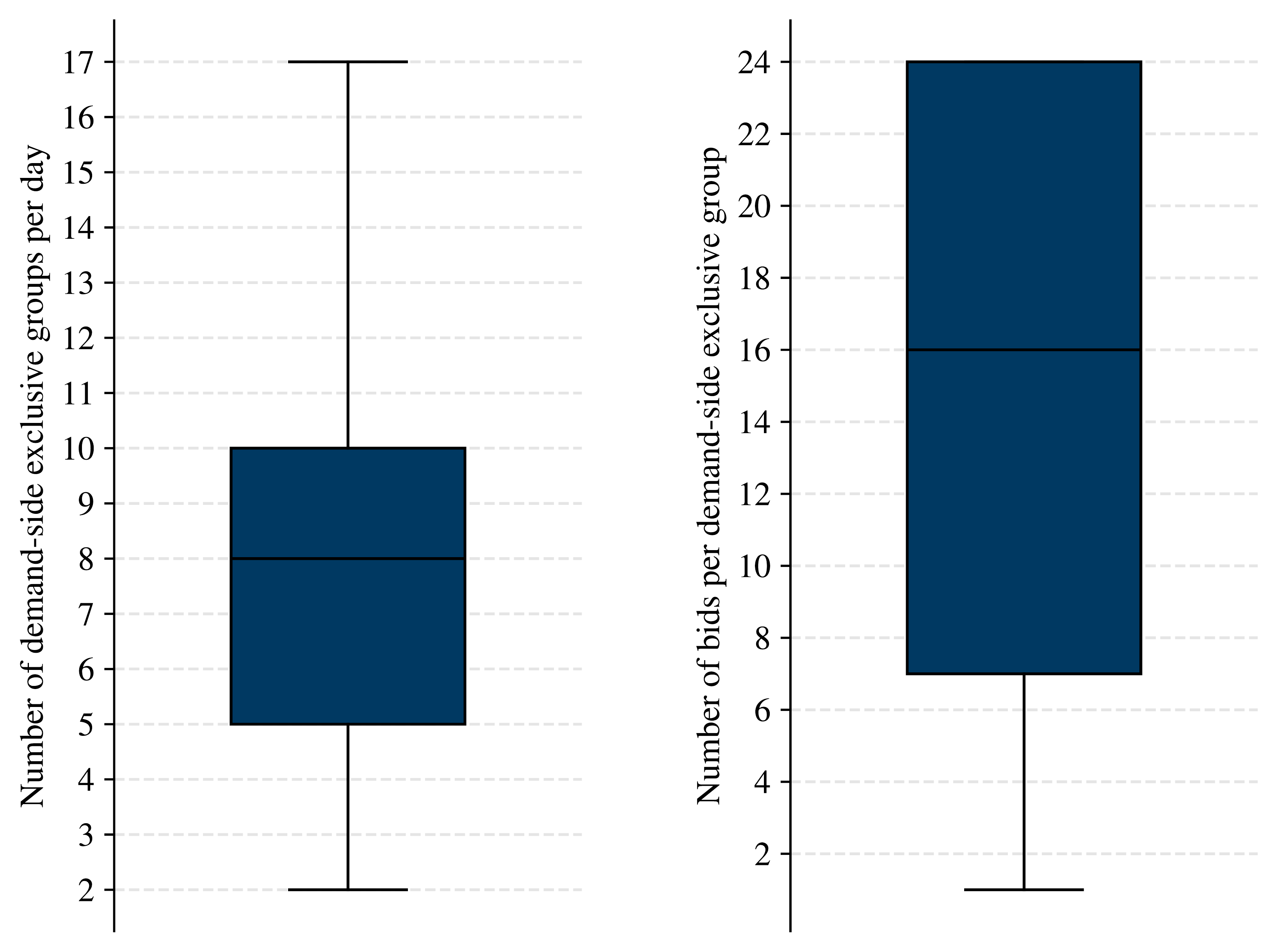}
    \end{subfigure}
    \hspace{1cm}
    \begin{subfigure}[c]{0.45\textwidth} 
        \centering
        \caption{Supply bids.}
        \label{subfig: supply number of bids}
        \includegraphics[height=5cm]{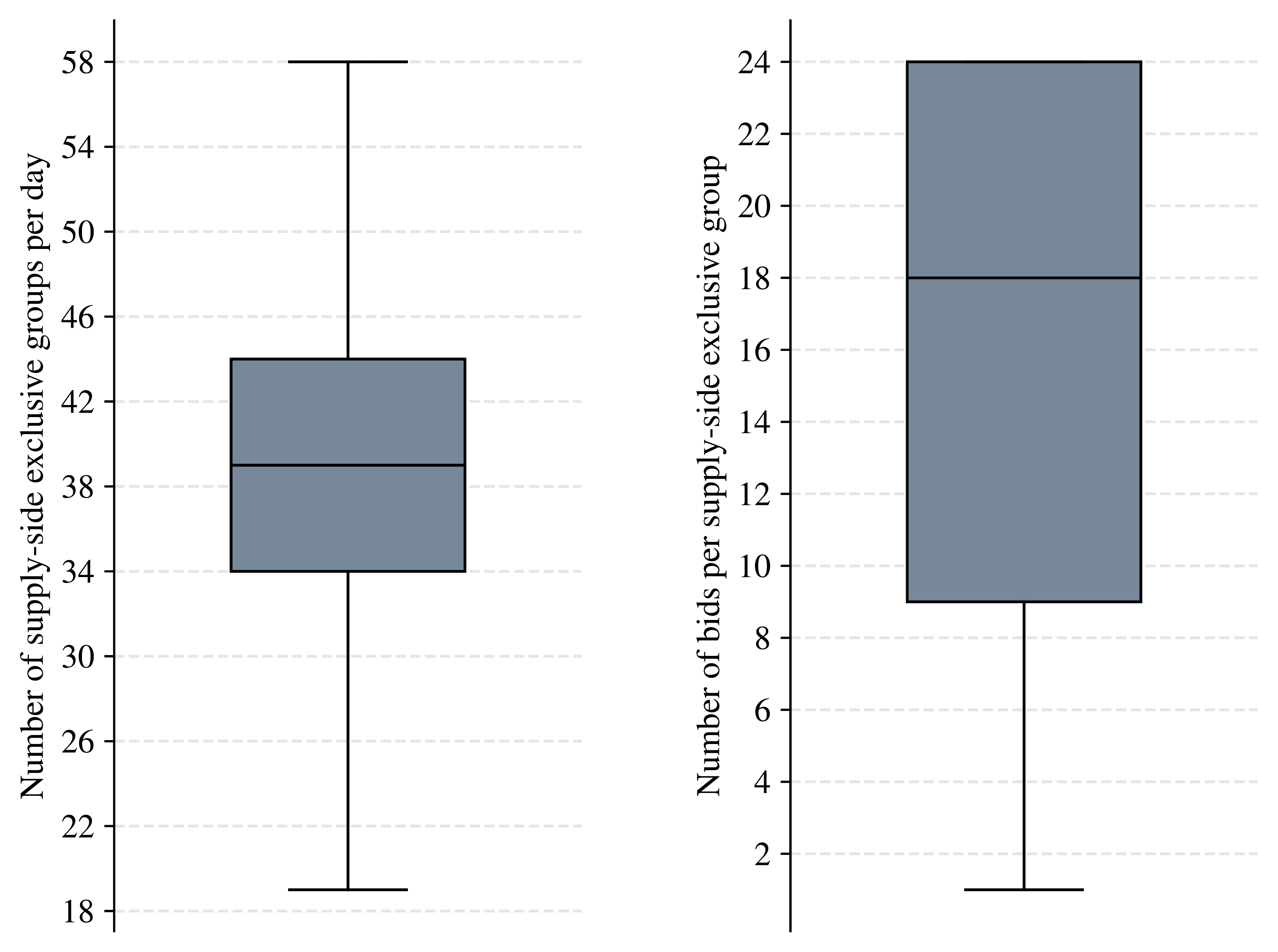}
    \end{subfigure}
\label{fig: number of bids}
\end{figure}

It is crucial to highlight that day-ahead and intraday auctions are \textit{forward} markets, requiring participants to account for \textit{opportunity costs} and risk preferences in their bidding strategies. Opportunity costs may arise naturally within the market or from design flaws, as exemplified by INC-DEC gaming strategies~\citep{graf2020simplified}.
In this work, we do not address the estimation of opportunity costs or the handling of risk preferences. Instead, we assume risk-neutral participants who aim to maximize the expected profit based on a generic cost or valuation function\endnote{A body of literature focuses on estimating opportunity costs, particularly in the context of hydropower; for further references, see \cite{lohndorf2023value}. In Appendix~\ref{app: conditional value at risk}, we study risk-averse participants maximizing the conditional value at risk (CVaR) .} 

That said, the forward nature of electricity auctions requires participants to make rapid decisions as updated weather forecasts continually refine information about future intraday and real-time electricity prices.\endnote{Moreover, there is a discussion about increasing the frequency of intraday auctions to enhance the overall efficiency of the European market, with a proposal to implement auctions every 15 minutes~\citep{graf2024frequent}.} This urgency demands that XOR bid selection methods be computationally efficient to avoid acting on outdated information. Despite this, \cite{karasavvidis2024optimal} formulated the bilevel model as a mixed-integer linear program (MILP), where binary variables grow linearly with the number of XOR bids and price scenarios. As a result, in one of their case studies, the model failed to terminate in an hour when determining only ten bids and considering just 20 price scenarios.

Our work introduces a linear programming (LP) approach for determining XOR bids. The \textit{polynomial time complexity} of our method is crucial for two reasons: First, increasing uncertainty in electricity prices requires considering more price scenarios for robustness. Second, as we will argue later, the historical limit of 24 bids is likely outdated, with the number of allowable package bids expected to rise.
We show that the stochastic bilevel model, maximizing expected profit under no market power and equilibrium assumptions, can be reformulated as an LP exploiting a \textit{total unimodular constraint matrix} when the participant has a finite set of packages. If not, a finite subset of bids can be pre-selected by some heuristics.

When agents bid according to this algorithm, lack market power, and the auction accepts their most profitable bids, we can derive an upper bound on the profit loss caused by the limit on the number of bids.
Moreover, assuming competitive equilibrium, we can extend this to bound the welfare loss for the entire auction.
In particular, our bound is based on the \textit{Wasserstein distance} between two probability measures: a discrete measure $\mathcal{Q}$ derived from the scenarios used for bid selection and a possibly continuous measure $\mathcal{P}$ representing the ``true'' price distribution. This relationship provides an interesting interpretation of the problem of bid selection and missing bids as a case of approximation of one probability measure with another~(\Cref{fig: probability distributions}) - a well-explored topic in stochastic programming~\citep{rujeerapaiboon2022scenario}.

\begin{figure}[t]
\caption{An illustration of distributions $\mathcal{P}$ and their discrete approximations $\mathcal{Q}$.}
\centering
    \begin{subfigure}[b]{0.4\textwidth} 
        \caption{High uncertainty.}
        \label{fig: probability_distribution_high}
        \includegraphics[width=\textwidth]{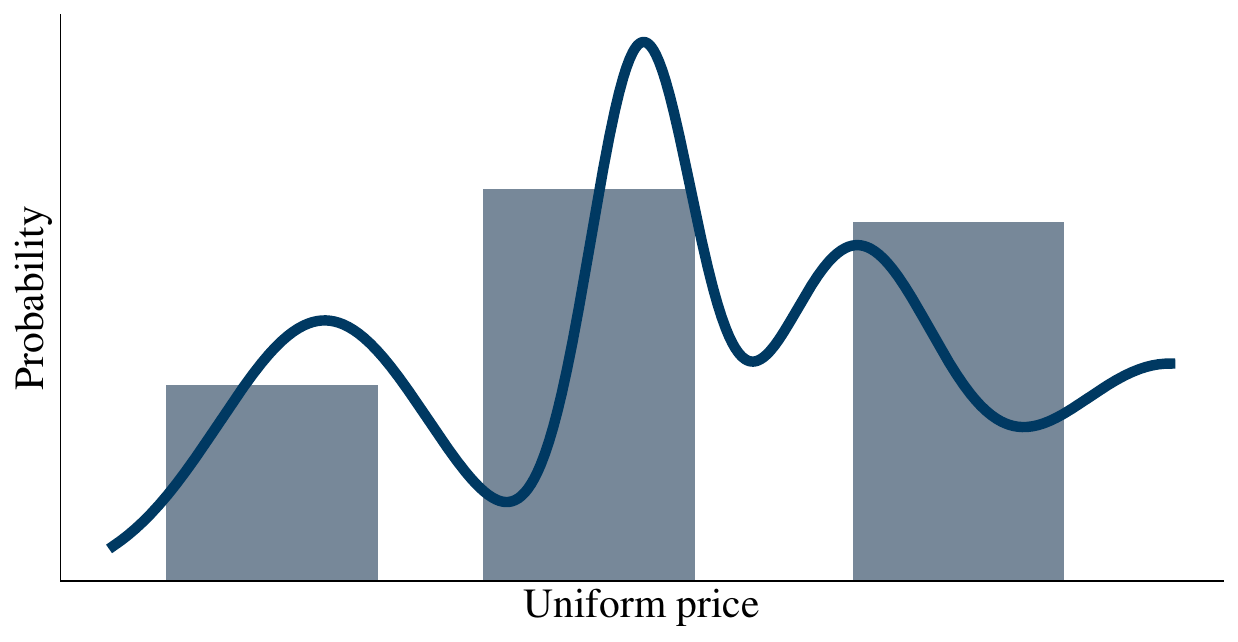}
    \end{subfigure}
    \hspace{0.05\textwidth}
    \begin{subfigure}[b]{0.4\textwidth} 
        \caption{Low uncertainty.} 
        \label{fig: probability_distribution_low}
        \includegraphics[width=\textwidth]{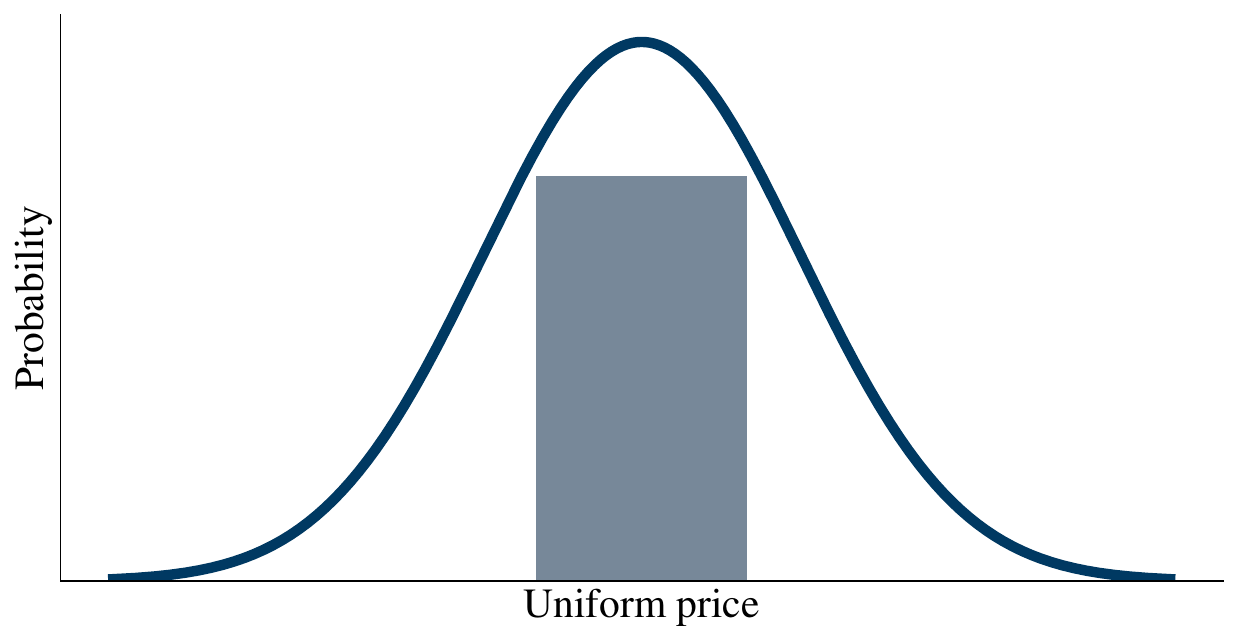}
    \end{subfigure}

{\footnotesize \textit{Notes:} When uncertainty is high, $\mathcal{P}$ is more dispersed, requiring $\mathcal{Q}$ to use a larger number of scenarios for a good approximation. In contrast, when uncertainty is low, $\mathcal{P}$ is concentrated around a single mode, allowing $\mathcal{Q}$ to approximate it with fewer scenarios.}
\label{fig: probability distributions}
\end{figure}

This theorem captures a key concept in \textit{microeconomic theory}: In markets where equilibrium prices are known or can be easily determined through decentralized bargaining, no auctioneer is needed, as prices alone can guide allocations~\citep{mas1995microeconomic}. In such cases, a single ``bid'' - indicating how much an agent is willing to produce or consume - is sufficient. However, when prices are unknown or difficult to determine, the auctioneer requires as much information as possible about agents' preferences to guide efficient allocations~\citep{milgrom2017discovering}.
Electricity markets lie somewhere in between these extremes. Although electricity prices are not fully known in advance, they usually follow predictable patterns~\citep{lago2021forecasting}. 
Therefore, although a single self-schedule based on forecasted prices might lead to inefficient or infeasible allocations, providing the auctioneer with detailed preferences might also be unnecessary if good price signals are available.

In general, equilibrium prices cannot be guaranteed in nonconvex markets, though near-equilibrium prices can often be approximated (see, e.g., \cite{starr1969quasi}, \cite{bichler2018matter}, or \cite{hubner2025approximate}). However, unlike electricity auctions, agents in most other auctions typically lack prior knowledge of these prices. Consequently, studies addressing the missing bids problem - such as \cite{goetzendorff2015compact} for TV ads and \cite{bichler2023taming} for spectrum auctions - advocate for \textit{parametric bids}, which work as follows: The agent submits parameters $\theta_1, \ldots, \theta_n$, which the auctioneer uses to derive a set $\mathcal{M}(\theta_1, \ldots, \theta_n)$ of packages the agent is interested in, along with a function $f(\theta_1, \ldots, \theta_n): \mathcal{M}(\theta_1, \ldots, \theta_n) \to \mathbb{R}$ that assigns bid prices to these packages, eliminating the need to enumerate all possible packages. These bids, commonly used for thermal plants and battery storage in US power systems, are known as ``multi-part bids'' and typically include parameters such as start-up costs or storage capacity~\citep{herrero2020evolving}. A more direct approach would allow agents to submit sets $\mathcal{M}$ and functions $f$ directly, bypassing intermediate parameters. However, the auctioneer must restrict the complexity of $f$ and $\mathcal{M}$ to ensure tractable optimization, making this approach suitable only for agents with sufficiently simple preferences.\endnote{The proposal by \cite{bobo2021price} suggests allowing agents to bid ``state variables'' and matrices defining a polyhedral set $\mathcal{M}$ and a convex piecewise-linear function $f$. This effectively allows any LP-representable $\mathcal{M}$ and $f$. To handle nonconvexities, the framework could extend to MILP-representable $\mathcal{M}$ and $f$.}

The economics of most thermal power plants and battery storage systems can typically be captured through computationally feasible parametric or functional bids~\citep{herrero2020evolving}. However, this may not apply to all resources, particularly emerging demand-side agents aggregating decentralized energy sources.\endnote{Electrical energy is a complex commodity. Consequently, the systems that produce it, such as large thermal or hydropower plants, are intricate techno-economic systems that cannot always be easily described by a simple set of parameters or functional bid. On the demand side, human behavior further complicates matters, and aggregators such as utilities must account for the dynamics of their entire retail system when submitting buy bids in the wholesale market. 
The potential difficulty of parametric bids in accurately capturing preferences might be one reason why a large number of agents (about 40\%, according to \cite{byers2022auction}, for PJM in 2018) choose to self-schedule instead of using the available parametric bid formats.}
Another challenge is the willingness to participate. Complex bidding procedures often discourage participation, either due to a lack of understanding or a reluctance to share detailed operational data~\citep{milgrom2017discovering}.
XOR bids offer a simpler alternative. They are easy to understand, require no operational data, and are applicable across all technologies. Their simplicity and versatility make them an ideal complement for agents for whom parametric and functional bids are unsuitable.

This article aims to help auction designers and agents using XOR bids understand and reduce the welfare loss caused by missing bids. To achieve this, \Cref{sec: bid selection} introduces our bid selection algorithm. \Cref{sec: welfare loss} presents our Wasserstein theorem. In \Cref{sec: simulations}, we apply the algorithm to simulate the bidding behavior of a generator, battery storage, and a utility serving load in the German bidding zone during 2023. \Cref{sec: disucsions} demonstrates that package bids are not necessarily nonconvex, discusses aggregator and portfolio bidding, and briefly debates other types of package bids. Finally, \Cref{sec: conclusion} concludes the article.

\section{Optimal Package Bid Selection}\label{sec: bid selection}

Consider an auction with $T$ periods, indexed by $t \in \mathcal{T} = \{1, \ldots, T\}$, and a participant with a valuation function $v: \mathbb{R}^T \to \mathbb{R} \cup \{-\infty\}$. The function $v(x)$ quantifies the agent's cost $\big(v(x) < 0\big)$ or utility $\big(v(x) > 0\big)$ associated with the consumption or supply of power profile $x \in \mathbb{R}^T$. If $v(x) = -\infty$, the profile $x$ is infeasible for the agent, possibly due to technical constraints.

The agent may submit $B$ package bids in an exclusive group, where each bid is defined by $(x_b, p_b)$ for $b \in \mathcal{B} = \{1, \ldots, B\}$. Here, $x_b \in \mathbb{R}^T$ is a power profile of interest, and $p_b \in \mathbb{R}$ is the price the agent is willing to pay (or expects to receive) for $x_b$. The value $x_{bt}\in\mathbb{R}$ denotes the power consumed from the grid ($x_{bt} > 0$) or supplied to the grid ($x_{bt} < 0$) during period $t$, typically measured in MW. Bid prices $p_b$ are expressed in euros (€) in European day-ahead and intraday auctions.

\subsection{Stochastic Bilevel Program}\label{subsec: stochastic bilevel program}

The agent determines XOR bids under the assumptions of no market power, no perfect information, and the expectation of an equilibrium outcome. This implies that the agent forecasts price scenarios $\lambda_s \in \mathbb{R}^T, s \in \mathcal{S} = \{1, \ldots, S\}$ with probabilities $\pi_s$, and assumes the auctioneer selects the most profitable bid in each scenario:
\begin{align} \label{eq: most profitable bid}
    & (x_s^\ast,p_s^\ast) \in \argmax \big\{p_b - \langle \lambda_s, x_b \rangle \ \big| \ (x_b, p_b) \in \nonumber \\
    & \qquad \{ (\mathbf{0},0), (x_1,p_1), \ldots, (x_B,p_B)\} \big\} \quad \forall s\in\mathcal{S} .
\end{align}
Here, $\langle \cdot, \cdot \rangle$ denotes the scalar product. Specifically, $\langle \lambda_s, x_b \rangle = \sum_{t \in \mathcal{T}} \lambda_{st} \cdot x_{bt}$ represents the payment the agent receives (or pays) if bid $(x_b, p_b)$ is accepted under the uniform pricing rule.
The XOR condition ensures that the auctioneer accepts at most one bid from the group, or none at all - defaulting to the zero bid $(\mathbf{0}, 0)$, where $\mathbf{0} = (0, \ldots, 0)$ if all bids are unprofitable.

We assume a risk neutral agent aiming to determine $B$ package bids $(x_b, p_b), b \in \mathcal{B}$, which maximize the expected profit. To do so, the agent seeks to solve:\endnote{We assume the existence of an optimum, and therefore use $\max$ instead of $\sup$. However, it is important to note that this assumption does not necessarily hold, as no additional properties of $v$ have been specified.}
\begin{align}\label{eq: stochastic bilevel program}
    \max_{x_b,p_b} & \quad \sum_{s\in\mathcal{S}} \pi_s \cdot \big( v(x_s^\ast) - \langle \lambda_s, x_s^\ast \rangle \big) \nonumber \\
    \text{s.t.} & \quad \eqref{eq: most profitable bid} \;\;\; \text{and} \;\; (x_b,p_b)\in\mathbb{R}^{T+1} .
\end{align}
This is a stochastic bilevel program, with the lower level defined by \eqref{eq: most profitable bid}. It represents the ``optimistic'' version, where the upper level optimizes over the lower-level outcome $x_s^\ast$~\citep{kleinert2021survey}. In reality, if there are multiple most profitable bids, the auctioneer, not the agent, makes the selection. 
The MILP formulation in \cite{karasavvidis2024optimal} can be derived from~\eqref{eq: stochastic bilevel program} using standard bilevel reformulation techniques~\citep{kleinert2021survey}. In the following, we show that this bilevel program can be reformulated as an LP.

\subsection{Truthful Bidding}\label{subsec: truthful bidding}

As discussed in the Introduction, a fundamental result in economic theory asserts that a mechanism leading to a Walrasian equilibrium guarantees truthful bidding as a dominant strategy, provided agents cannot influence the uniform price~\citep{mas1995microeconomic, azevedo2019strategy}. We show that this strategy-proofness extends to our setting, where agents cannot bid their valuation $v$ directly but must choose from a set of XOR package bids. To demonstrate this, we first analyze how program~\eqref{eq: stochastic bilevel program} would be affected by truthful bidding and provide proof that this is indeed a dominant strategy at the end of the section.

Assuming the agent bids truthfully, which means that they select a bid price equal to their valuation, that is, $p_b = v(x_b)$ for all $b \in \mathcal{B}$, we can eliminate $p_b$ as a decision variable from problem~\eqref{eq: stochastic bilevel program}. Consequently, the problem can be reformulated as:
\begin{align}\label{eq: stochastic bilevel program I}
    \max_{x_b\in\mathbb{R}^{T}} & \quad \sum_{s\in\mathcal{S}} \pi_s \cdot \big( v(x_s^\ast) - \langle \lambda_s, x_s^\ast \rangle \big) \nonumber \\
    \text{s.t.} & \quad x_s^\ast \in \argmax \big\{v(x_b) - \langle \lambda_s, x_b \rangle \ \big| \ \nonumber \\
    & \qquad x_b \in \{\mathbf{0}, x_1, \ldots, x_B\} \big\} \quad \forall s\in\mathcal{S} .
\end{align}
The objective function of the lower level now aligns with that of the upper level. As a result, the lower-level constraint can be omitted since the upper-level optimization will implicitly ensure the appropriate selection of $x_s^\ast$~\citep{kleinert2021survey}:
\begin{align}\label{eq: stochastic bilevel program II}
    \max_{x_b\in\mathbb{R}^{T}} & \quad \sum_{s\in\mathcal{S}} \pi_s \cdot \big( v(x_s^\ast) - \langle \lambda_s, x_s^\ast \rangle \big) \nonumber \\
    \text{s.t.} & \quad x_s^\ast \in \{\mathbf{0}, x_1, \ldots, x_B\} \quad \forall s\in\mathcal{S} .
\end{align}
With these considerations in place, we can now establish that truthful bidding, i.e., $p_b = v(x_b)$ for all $b \in \mathcal{B}$, is indeed a dominant strategy. As a consequence of this proposition, the agent can solve the single-level problem~\eqref{eq: stochastic bilevel program II} instead of the bilevel problem~\eqref{eq: stochastic bilevel program} without any loss of expected profit.

\begin{proposition} \label{proposition: truthful bidding}
Let $v^1$ and $v^2$ be the optimal values of program~\eqref{eq: stochastic bilevel program} without and with the constraints $p_b = v(x_b)$ for all $b \in \mathcal{B}$, respectively. Then, $v^2 \geq v^1$.
\end{proposition}
\begin{proof}
Let $v^3$ be the optimal value of program~\eqref{eq: stochastic bilevel program II}. As shown above, $v^3 = v^2$. Additionally, the feasible region for $x_s^\ast$ is larger in program~\eqref{eq: stochastic bilevel program II} than in~\eqref{eq: stochastic bilevel program}, where the lower-level constraint limits it. Thus, $v^3 \geq v^1$, thereby concluding the proof. 
\end{proof}

\subsection{More Bids than Scenarios} \label{subsec: heuristic}

Assume the agent is permitted to place more XOR package bids than the number of predicted price scenarios, that is, $B \ge S$. In this case, no optimization is necessary, as a solution to~\eqref{eq: stochastic bilevel program II}, and consequently to~\eqref{eq: stochastic bilevel program}, can be obtained simply by executing the following algorithm:
\begin{align}\label{eq: heuristic I}
    &\begin{array}{l}
    \text{(i) Forecast $S$ price scenarios $\lambda_s$.} \\ 
    \text{(ii) For each scenario, determine an optimal} \\ 
    \qquad \text{profile $\bar{x}_s \; \in \ \argmax \; v(x) - \langle \lambda_s, x \rangle .$} \\
    \text{(iii) Submit the package bids} \\
    \qquad\; \text{$\big \{ \big(\bar{x}_s,v(\bar{x}_s)\big) \ \big | \ s=1,\ldots,S \big\}$.}
    \end{array}
\end{align}
Our simulations in~\Cref{sec: simulations} show that using many price scenarios is beneficial in bid selection only when they are high quality and add meaningful information. For example, if only 20 out of 100 scenarios are reliable, focusing on these 20 is more effective than taking all 100 into account. In such cases,~\eqref{eq: stochastic bilevel program} can be efficiently solved using~\eqref{eq: heuristic I}.
Next, we consider another case where solving~\eqref{eq: stochastic bilevel program} is tractable. 

\subsection{Finitely Many Packages} \label{subsec: LP reformulation}

Assume the agent is interested in only a finite number of packages. That is, there exists a finite set $\Sigma = \{\bar{x}_1, \ldots, \bar{x}_K\}$ of profiles for which $v(x) \neq -\infty$, and $v(x) = -\infty$ for all $x \in \mathbb{R}^T \setminus \Sigma$. In this case, problem~\eqref{eq: stochastic bilevel program II} is given by:
\begin{align}\label{eq: stochastic bilevel program III}
    \max_{x_b} & \quad \sum_{s\in\mathcal{S}} \pi_s \cdot \big( v(x_s^\ast) - \langle \lambda_s, x_s^\ast \rangle \big) \nonumber \\
    \text{s.t.} & \quad x_s^\ast \in \big\{\mathbf{0}, x_1, \ldots, x_B \big\} \quad \forall s\in\mathcal{S} \nonumber \\
    & \quad x_b \in \big\{\bar{x}_1,\ldots,\bar{x}_k, \ldots, \bar{x}_K\big\} \quad \forall b\in\mathcal{B} .
\end{align}
The logic behind this program is as follows: First, the agent selects $B$ packages $x_b$ from the set $\{\bar{x}_1, \ldots, \bar{x}_K\}$. Second, the agent chooses $S$ packages $x_s^\ast$ from the selected packages $\{\mathbf{0}, x_1, \ldots, x_B\}$ to represent the most profitable package in scenario $s$.

This logic can be captured using two types of binary variables. Let $\delta_k \in \{0, 1\}$ denote whether the agent selects package $\bar{x}_k$ as one of its XOR bids, and let $\gamma_{ks} \in \{0, 1\}$ denote whether package $\bar{x}_k$ is chosen as the most profitable package $x_s^\ast$. Importantly, $\gamma_{ks}$ can only be 1 if $\bar{x}_k$ is selected as one of the XOR bids, that is, $\gamma_{ks} \leq \delta_k$.

With these definitions, we can formulate~\eqref{eq: stochastic bilevel program III} as the following binary program:
\begin{subequations}\label{eq: binary program}
\begin{align} 
\max_{\delta, \gamma} & \quad \sum_{s\in\mathcal{S}} \pi_s \cdot \big( v(x_s^\ast) - \langle \lambda_s, x_s^\ast \rangle \big) \label{eq: binary program 0} \\
\text{s.t.} & \quad x_s^\ast = \sum_{k\in\mathcal{K}} \gamma_{ks} \cdot \bar{x}_k \quad \forall s \in \mathcal{S} \label{eq: binary program 1} \\
& \quad \sum_{k \in \mathcal{K}} \gamma_{ks} \le 1 \quad \forall s \in \mathcal{S} \label{eq: binary program 2} \\ 
& \quad \gamma_{ks} \le \delta_{k} \quad \forall s \in \mathcal{S}, \: k\in\mathcal{K} \label{eq: binary program 3} \\
&  \quad \sum_{k \in \mathcal{K}} \delta_k = B \label{eq: binary program 4} \\
& \quad \delta_{k},\gamma_{ks}\in\{0,1\} \quad \forall s \in \mathcal{S}, \: k\in\mathcal{K} \label{eq: binary program 5} .
\end{align}
\end{subequations}
Here, $\mathcal{K} = \{1, \ldots, K\}$.
Constraints~\eqref{eq: binary program 2}–\eqref{eq: binary program 4} implement the selection logic of program~\eqref{eq: stochastic bilevel program III}. Specifically, \eqref{eq: binary program 4} ensures that the agent selects exactly $B$ packages to bid on, \eqref{eq: binary program 3} guarantees that only packages chosen by the agent can be selected as the most profitable, and \eqref{eq: binary program 2} ensures that at most one package is selected as the most profitable in each scenario.

To simplify this program further, we can eliminate $x_s^\ast$ by substituting~\eqref{eq: binary program 1} into the objective function. Doing this, we obtain $v(\sum_{k \in \mathcal{K}} \gamma_{ks} \cdot \bar{x}_k)$. Since at most one $\gamma_{ks}$ is 1 and the others are 0, we can draw the sum out of the function, resulting in $\sum_{k \in \mathcal{K}} \gamma_{ks} \cdot v(\bar{x}_k)$. As $\bar{x}_k$ is a parameter of the optimization model, we can define $\bar{v}_k = v(\bar{x}_k)$ as well as a parameter. Thus, we can formulate~\eqref{eq: binary program} as:
\begin{align}\label{eq: binary program I}
\max_{\delta, \gamma} & \quad \sum_{s \in \mathcal{S}} \pi_s \cdot \sum_{k \in \mathcal{K}} \big( \bar{v}_k - \langle \lambda_s, \bar{x}_{k} \rangle \big) \cdot \gamma_{ks} \nonumber \\
\text{s.t.} & \quad \eqref{eq: binary program 2} - \eqref{eq: binary program 5} .
\end{align}
By examining the constraint matrix formed by constraints \eqref{eq: binary program 2}–\eqref{eq: binary program 5}, it becomes clear that the binary restrictions are redundant. Therefore, the binary variables can be relaxed to continuous variables $\delta_{k}, \gamma_{ks} \in [0, 1]$ without sacrificing optimality.

\begin{proposition} \label{proposition: lp relaxation}
    Program~\eqref{eq: binary program I} and its LP-relaxation are equivalent in the sense that their optimal values coincide, and there is a point $(\delta_{k}^\ast,\gamma_{ks}^\ast)_{s \in \mathcal{S}, k\in\mathcal{K}}$ which is an optimal solution to both.
\end{proposition}
\begin{proof}
It is straightforward to verify that the matrix formed by constraints~\eqref{eq: binary program 2}~-~\eqref{eq: binary program 4} is totally unimodular by partitioning the rows~\citep[cf. Theorem 2.7 in part III.1 of][]{nemhauserwolsey1988}.
Moreover, since adding identity matrices does not cancel total unimodularity, the polytope formed by~\eqref{eq: binary program 2}~-~\eqref{eq: binary program 4} and $\delta_k, \gamma_{ks}\in[0,1] \ \forall s \in \mathcal{S}, \: k\in\mathcal{K}$ is integral~\citep[cf. Proposition 2.1 and Theorem 2.5 in part III.1 of][]{nemhauserwolsey1988}. 
Consequently, there is an integral vertex $(\delta_k^\ast, \gamma_{ks}^\ast)_{s \in \mathcal{S}, k\in\mathcal{K}}$, which solves the LP relaxation and thus also \eqref{eq: binary program I}. 
\end{proof}

In Appendix~\ref{app: bilevel knapsack}, we highlight that our package bid selection problem~\eqref{eq: binary program I} is a special variant of the bilevel knapsack problem.
Interestingly, \Cref{proposition: lp relaxation} shows that this specific variant can be solved in polynomial time.
This is surprising, since - as discussed by \citet{caprara2014study} - classical knapsack problems are NP-complete, and bilevel problems built on top of NP-complete single-level problems are typically $\Sigma_2^p$-complete. This implies that they cannot be reformulated as single-level integer programs of polynomial size unless the polynomial hierarchy collapses. Moreover, \citet{caprara2014study} showed that many bilevel knapsack variants studied in the literature are indeed $\Sigma_2^p$-complete. Our variant is, therefore, exceptionally tractable.

While we assumed a risk-neutral agent maximizing expected surplus, \Cref{proposition: lp relaxation} remains valid as long as the agent optimizes an LP-representable risk measure like, for example, the Conditional Value at Risk (CVaR)~\citep{rockafellar2000optimization}. Appendix~\ref{app: conditional value at risk} provides a CVaR-based formulation, which, in our case, maximizes the expected worst-case profit.

\subsection{Infinitely Many Packages} \label{subsec: pre-selection}

We have identified two cases where~\eqref{eq: stochastic bilevel program} can be easily solved: when the number of bids allowed exceeds the number of scenarios, or when the agent is only interested in a finite set of packages. However, these conditions may not always hold, as electricity is a divisible good, and thus, agents could have an infinite number of packages they may be interested in.

The challenge in solving program~\eqref{eq: stochastic bilevel program} or its single-level equivalent~\eqref{eq: stochastic bilevel program II} arises from the inclusion of $v(x)$, which can be an arbitrarily complex function. For example, it could represent the cost function of thermal power plants with intricate unit commitment dynamics or the utility function of an electric vehicle fleet operator shaped by its internal revenue management problem.

When $S>B$ and the agent has infinitely many packages of interest, we adopt a hybrid approach combining both strategies:
\begin{align}\label{eq: heuristic II}
    &\begin{array}{l}
    \text{(i) and (ii) from \eqref{eq: heuristic I}.} \\ 
    \text{(iii) Solve~\eqref{eq: binary program I} to select $B$ packages} \\ 
    \qquad \text{$x_1,\ldots,x_B$ out of the $S$ packages $\bar{x}_1,\ldots,\bar{x}_S$.} \\
    \text{(iii) Submit the package bids} \\
    \qquad\; \text{$\big \{ \big(x_b,v(x_b)\big) \ \big | \ b=1,\ldots,B \big\}$.}
    \end{array}
\end{align}
Although this does not guarantee an optimal solution to~\eqref{eq: stochastic bilevel program}, it enables the agent to quickly determine package bids, even when $B$ and $S$ are large.

\section{Bound on the Welfare Loss} \label{sec: welfare loss}

Assume that no agent has market power and that there is an equilibrium price $\lambda$. If an agent could submit unlimited package bids, they could bid truthfully on every package:
\begin{equation}\label{eq: every package bid}
    \big \{ \big(x_b,v(x_b)\big)\in\mathbb{R}^{T+1} \ \big| \ v(x_b) \neq -\infty \big\}.
\end{equation}
Given the equilibrium, the auctioneer would select the most profitable package as per equation~\eqref{eq: most profitable bid}, allowing the agent to achieve their maximum profit:\endnote{Note that this is equivalent to the profit the agent could achieve with perfect knowledge of $\lambda$. As discussed in the introduction, an auctioneer is unnecessary if everyone knows the equilibrium price.}
\begin{equation*}
    \max_x \; v(x) - \langle \lambda, x \rangle .
\end{equation*}
However, agents can only submit a subset of the bids in equation~\eqref{eq: every package bid}, which may prevent them from achieving their maximal profit. In this section, we develop a bound on the profit loss caused by limiting the number of bids.

Given our equilibrium assumption, the bound on agents' surplus loss translates directly into a bound on the total welfare loss of the auction. This follows from the first welfare theorem in a partial equilibrium context, which states that when each agent maximizes their utility at the equilibrium price, total welfare is also maximized~\citep{mas1995microeconomic}. Therefore, the loss in welfare at price $\lambda$ is the sum of the surplus loss of all agents.

\subsection{Loss of Profit} \label{subsec: missed profit}

The profit loss depends on the agent's bidding strategy. We assume that the agent follows the simple strategy in~\eqref{eq: heuristic I} when fewer scenarios are considered than the allowed number of bids. This simplifies the theoretical analysis and serves as a lower bound compared to the more sophisticated strategies in~\eqref{eq: binary program I} or~\eqref{eq: heuristic II}. Given this strategy, the agent's profit loss is:
\begin{subequations}\label{eq: missed profit}
\begin{align}
& \Gamma(\lambda) \; = \; \max_x \; v(x) - \langle \lambda, x \rangle \label{eq: max profit} \\
& \qquad  - \; \max_x \; \big\{ v(x) - \langle \lambda, x \rangle \ \big| \ x\in\{\bar{x}_1, \ldots, \bar{x}_S\} \big\} , \label{eq: profit under restriction}
\end{align}
\end{subequations}
where~\eqref{eq: max profit} represents the agent’s profit if no bid restrictions exist, and~\eqref{eq: profit under restriction} is the profit achieved with strategy~\eqref{eq: heuristic I}. The profit loss $\Gamma(\lambda)$ depends on the equilibrium price $\lambda$. If one of the $S$ price scenarios in strategy~\eqref{eq: heuristic I} is correct, then $\Gamma(\lambda^\ast)=0$; otherwise, $\Gamma(\lambda)\ge0$.

\subsection{Ex Ante and Ex Post} \label{subsec: ex ante and ex post}

Ex ante, the equilibrium price $\lambda$ is uncertain and can be described by a probability space $(\Omega, \mathcal{F}, \mathcal{P})$, where $\Omega\subseteq\mathbb{R}^T$ is the sample space, $\mathcal{F}\subseteq2^\Omega$ is a $\sigma$-algebra of events and $\mathcal{P}$ is a probability measure on $\mathcal{F}$. 
Ex post, the realized price follows a \textit{degenerate} probability distribution, with $\mathcal{P}(\lambda)=1$ for a specific $\lambda \in \Omega$. The expected profit loss (ex-ante), as well as the actual profit loss (ex-post), are given by the expected value:
\begin{equation*} \label{eq: expected regret}
    \text{E} \big[ \Gamma(\lambda) \big] = \int_{\Omega} \Gamma(\lambda) \ \mathcal{P}(d\lambda) .
\end{equation*}
The agent does not necessarily have prior knowledge of the probability distribution $\mathcal{P}$. Instead, following the bidding strategy described in~\eqref{eq: heuristic I}, the agent considers $S$ distinct price scenarios $\lambda_s$, each of which occurs with a corresponding probability $\pi_s$. These scenarios collectively define a discrete probability distribution over the same measurable space $(\Omega, \mathcal{F})$:
\begin{equation*} \label{eq: measure Q}
\mathcal{Q}(A)=\sum_{s\in\mathcal{S}} \pi_s \cdot \bm{1}_{A}(\lambda_s)
\end{equation*}
where $A \in \mathcal{F}$ is an event and $\bm{1}_{A}(\lambda_s)$ the \textit{Dirac measure} indicating whether $\lambda_s$ is in $A$ or not. 

\subsection{Wasserstein Bound} \label{subsec: wasserstein bound}

To bound the profit loss $\Gamma(\lambda)$, we establish a relation between the probability measures $\mathcal{P}$ and $\mathcal{Q}$. To achieve this, we employ probability metrics, which are commonly used to quantify the similarity between two probability distributions~\citep[cf.][]{rachev2013}. Among these metrics, the Wasserstein distance is particularly notable. Its popularity comes from its deep ties to the theory of optimal transport and its intuitive interpretation as the minimal cost of transporting mass from $\mathcal{P}$ to $\mathcal{Q}$~\citep[cf.][]{villani2009}.
Specifically, the Wasserstein distance of order 1 between $\mathcal{P}$ and $\mathcal{Q}$ is defined as:
\begin{subequations}\label{eq: wasserstein distance}
\begin{align} 
d_W(\mathcal{P}, \mathcal{Q}) = & \\
\ \inf_{\pi} & \; \int_{\Omega\times\Omega} \| \lambda_1, \lambda_2 \|_2 \ \mu(d\lambda_1,d\lambda_2) \\
\ \text{s.t.} & \; \int_{\Omega} \mu(\lambda_1,d\lambda_2)  = \mathcal{P}(\lambda_1) \quad \forall\lambda_1\in\Omega  \\
&  \; \int_{\Omega} \mu(d\lambda_1,\lambda_2)  = \mathcal{Q}(\lambda_2) \quad \forall\lambda_2\in\Omega ,
\end{align} 
\end{subequations}
where $\| \cdot \|_2$ denotes the Euclidean norm. Using the Wasserstein distance $d_W(\mathcal{P},\mathcal{Q})$, we can give a bound on the profit loss as follows. The proof can be found in Appendix~\Ref{app: proof theorem wasserstein}. 

\begin{theorem} \label{theorem: wasserstein}
Let $L=2\cdot \max \big\{ \|x\|_2 \ \big | \ x\in\mathbb{R}^T, v(x) \neq -\infty \big\} $.
Then $\text{E}\big[ \Gamma(\lambda) \big] \le L \cdot d_W(\mathcal{P},\mathcal{Q})$. 
\end{theorem}

\Cref{theorem: wasserstein} highlights the interaction between price uncertainty and the number of bids $B$. When uncertainty is high, $\mathcal{P}$ is more dispersed, necessitating a larger number of bids $B$ for the agent to ``obtain'' a $\mathcal{Q}$ that closely approximates $\mathcal{P}$. Conversely, when uncertainty is low, $\mathcal{P}$ is more concentrated, allowing a good approximation even with fewer bids~(\Cref{fig: probability distributions}).

\section{Simulations: Generator, Storage, and Flexible Demand} \label{sec: simulations}

So far, we have discussed an abstract auction with $T$ commodities and generic agents defined by their valuation $v(x)$. Now, we shift to explicit simulations using the 2023 German bidding zone in the coupled European day-ahead auction, where the commodities represent the 24 hours of the upcoming day. We use historical electricity prices from the ENTSO-E transparency platform~\citep{entsoe}, assuming that they reflect competitive equilibrium prices, despite potential market power and paradoxically rejected bids.
To illustrate that XOR bids work well for all types of participants, we consider three exemplary agents: 
\begin{itemize}
    \item a \textbf{thermal generator} modeled with the classical unit commitment framework~\citep{carrion2006computationally, karasavvidis2021optimal}.
    \item a \textbf{battery storage system} based on its standard model~\citep{de2019implications, karasavvidis2021optimal}.
    \item a \textbf{district heating utility} employing heat pumps or gas boilers with demand-shifting capabilities through thermal storage, inspired by~\cite{bobo2021price} and utilizing heat load data from~\cite{ruhnau2019time}.
\end{itemize}
All agents are modelled with an infinite number of potential packages of interest. Consequently, we apply algorithm~\eqref{eq: heuristic II} to select XOR package bids. We then analyze the profit loss of this bidding strategy, performing sensitivity analyzes on the number of bids, the number of scenarios, and the quality of price information.

The code to replicate these experiments can be accessed online at \url{https://github.com/ThomasHubner/Package_Bids_Electricity_Auctions}. 
All experiments were run on a laptop with 32 GB RAM and an Intel i7-1260P processor. Optimization models were solved with Gurobi version 11 in standard settings. The execution time for Algorithm~\eqref{eq: heuristic II} was below 50 seconds for all cases considered.

\subsection{Case Studies} \label{subsec: case study description}

\paragraph{Thermal generator}
The plant can generate power within its minimum stable and maximum generation capacities and sell the output in the day-ahead auction. It can ramp up and down between hours, subject to a maximum ramping rate. Production costs consist of marginal costs, start-up and shutdown costs, and no-load costs. Additionally, when the plant is shut down (or started up), it must remain off (or on) for a predefined number of hours before it can start up (or shut down) again. A detailed model is provided in Appendix~\ref{app: thermal generator}, and the associated parameters are listed in~\Cref{table: thermal generator}.

\paragraph{Battery storage system} The battery can charge and discharge power, allowing it to both buy and sell in the day-ahead auction. Its state of charge must stay within its limits, and both charging/discharging efficiencies are considered. We assume that the battery's state of charge at the end of the day must be the same as at the beginning of the day.
The detailed model is given in Appendix~\ref{app: battery}, with the parameters listed in~\Cref{table: battery storage}.

\paragraph{District heating utility} We consider a utility that needs to supply a heat load using electric boilers, thus purchasing power in the day-ahead auction. The utility has aggregated heat storage available to shift the load between hours. It can also use a gas boiler if electricity prices are too high or reduce the load based on customer contracts, though it is not compensated for any curtailed load.
The detailed model is provided in Appendix~\ref{app: district heating utility}, with the parameters listed in~\Cref{table: district heating utility}.

\begin{table}[!htb]
    \footnotesize
    \centering
    \caption{Parameters used in the three case studies.}
    \label{table: case study data}
    \begin{subtable}{0.45\linewidth}
        \centering
        \caption{Thermal generator}
        \label{table: thermal generator}
        \begin{tabular}{|c|c|}
            \hline
            \textbf{Description} & \textbf{Value} \\
            \hline
            No-load cost (€/h) & 10,000 \\
            \hline
            Marginal cost of 0-200 MW (€/MWh) & 70 \\
            \hline
            Marginal cost of 200-400 MW (€/MWh) & 90 \\
            \hline
            Marginal cost of 400-600 MW (€/MWh) & 120 \\
            \hline
            Start-up cost (€/start-up) & 4,000 \\
            \hline
            Shut-down cost (€/shut-down) & 3,000 \\
            \hline
            Minimum stable generation (MW) & 100 \\
            \hline
            Maximum power output (MW) & 600 \\
            \hline
            Ramp-up rate (MW/h) & 200 \\
            \hline
            Ramp-down rate (MW/h) & 200 \\
            \hline
            Minimum up-time (h) & 4 \\
            \hline
            Minimum downtime (h) & 4 \\
            \hline
            Initial operating state (MW) & 0 \\
            \hline
            Initial off hours (h) & 0 \\
            \hline
        \end{tabular}
    \end{subtable}
    \hfill    
    \begin{subtable}{0.45\linewidth}
        \centering
        \caption{District heating utility}
        \label{table: district heating utility}
        \begin{tabular}{|c|c|}
            \hline
            \textbf{Description} & \textbf{Value} \\
            \hline
            Efficiency electric boiler & 1 \\
            \hline
            Efficiency gas boiler & 0.9 \\
            \hline
            Storage losses (\%/h) & 1 \\
            \hline
            Cost gas (€/MWh) & 90 \\
            \hline
            Value of served load (€/MWh) & 120 \\
            \hline
            Capacity electric boiler (MW) & 30 \\
            \hline
            Capacity gas boiler (MW) & 10 \\
            \hline
            Capacity heat storage (MWh) & 40 \\
            \hline
            Maximum storage charging limit (MW) & 20 \\
            \hline
            Maximum storage discharging limit (MW) & 20 \\
            \hline
            Initial state-of-charge storage (MWh) & 0 \\
            \hline
            \multicolumn{2}{|c|}{\underline{Hourly Heat Load (MW):} 19, 20, 20, 21, 24,} \\
            \multicolumn{2}{|c|}{32, 38, 36, 36, 35, 33, 32, 31, 31, } \\
            \multicolumn{2}{|c|}{31, 32, 33, 33, 33, 33, 32, 29, 23, 20} \\
            \hline
        \end{tabular}
    \end{subtable}
    
    \vspace{0.7cm}
    
    \begin{subtable}[b]{0.45\linewidth}
        \centering
        \caption{Battery storage system}
        \label{table: battery storage}
        \begin{tabular}{|c|c|}
            \hline
            \textbf{Description} & \textbf{Value} \\
            \hline
            Maximum charging limit (MW) & 10 \\
            \hline
            Maximum discharging limit (MW) & 10 \\
            \hline
            Charging efficiency & 0.9 \\
            \hline
            Discharging efficiency & 0.9 \\
            \hline
            Minimum feasible state-of-charge (MWh) & 0 \\
            \hline
            Maximum feasible state-of-charge (MWh) & 20 \\
            \hline
            Initial state-of-charge (MWh) & 10 \\
            \hline
        \end{tabular}
    \end{subtable}
\end{table}

\subsection{Probabilistic Price Forecasts} \label{subsec: price forecasts}

The quality of bid selection depends on the accuracy of the forecasted price scenarios. We use a common two-step approach to generate these scenarios, widely adopted in both literature and industry~\citep{lohndorf2013optimizing}, and shown to perform competitively~\citep{grothe2023point}. In the first step, we generate a point forecast for the price. In the second, we apply a post-processing method using the residuals from previous days to create the scenarios.

For the point forecast $y_d$ of the electricity price $\lambda_d$ on day $d$, we use an autoregressive model with LASSO regularization, implemented in the open-source Python package \textit{epftoolbox}~\citep{lago2021forecasting}. This model is trained on German price data from 1/1/2019 to day $d$.

Using the point forecast $y_d$, we generate $S$ price scenarios $\lambda_{ds}$ with equal probability $\pi_{ds} = \frac{1}{S}$ by:
\begin{subequations} \label{eq: scenario generation}
\begin{align}
&\lambda_{d1} = y_d, \\
&\lambda_{d2} = y_d - (y_{d-1} - \lambda_{d-1}), \\
&\lambda_{ds} = y_d - (y_{d-s-1} - \lambda_{d-s-1}) ,
\end{align}
\end{subequations}
where $(y_{d-1} - \lambda_{d-1})$ represents the forecast residual of the previous day. This approach generates $S$ price scenarios by adding the residuals of the last $S-1$ days to the point forecast $y_d$.

Since this method is simple to implement and publicly available, it serves as a lower bound on the price information accessible to bidders, who may use more advanced commercial forecasting methods.

\subsection{Number of Bids} \label{subsec: results bids}

In the first set of experiments, we examine the profit loss as a function of the number of XOR package bids submitted by the agent. Using algorithm~\eqref{eq: heuristic II}, we consider $S=150$ scenarios and bid numbers $B=1, 2, 5, 10, 20, 50, 100$. The case $B=1$ can be interpreted as a ``self-schedule'', where the agent provides the auctioneer with only a single inflexible schedule.\endnote{The term ``self-schedule'' is primarily used in US auctions to refer to an alternative to parametric bid formats, where asset owners, rather than the auctioneer's algorithm, are responsible for dispatching the resource. Specifically, when parametric bids were introduced for battery storage systems, ISOs adopted a ``self-schedule model'' for these systems, allowing storage operators to directly manage dispatch. This approach was implemented in response to the potential inaccuracies of parametric bid formats for storage~\citep{herrero2020evolving}.} To ensure robust results, we randomly sample 100 days from 2023.

\Cref{fig: results bids} shows the average profit achieved with algorithm~\eqref{eq: heuristic II}, expressed as a percentage of the maximum achievable profit~\eqref{eq: max profit} based on the realized uniform price. Assuming equilibrium prices, the auctioneer accepts the most profitable bid as defined by~\eqref{eq: most profitable bid}. Missed profits in \% are indirectly given by 100\% minus the profit achieved presented in~\Cref{fig: results bids}.

\begin{figure}[t]
\caption{Average profit by number of bids for 100 random days in 2023.}
\centering
\includegraphics[width=0.6\textwidth]{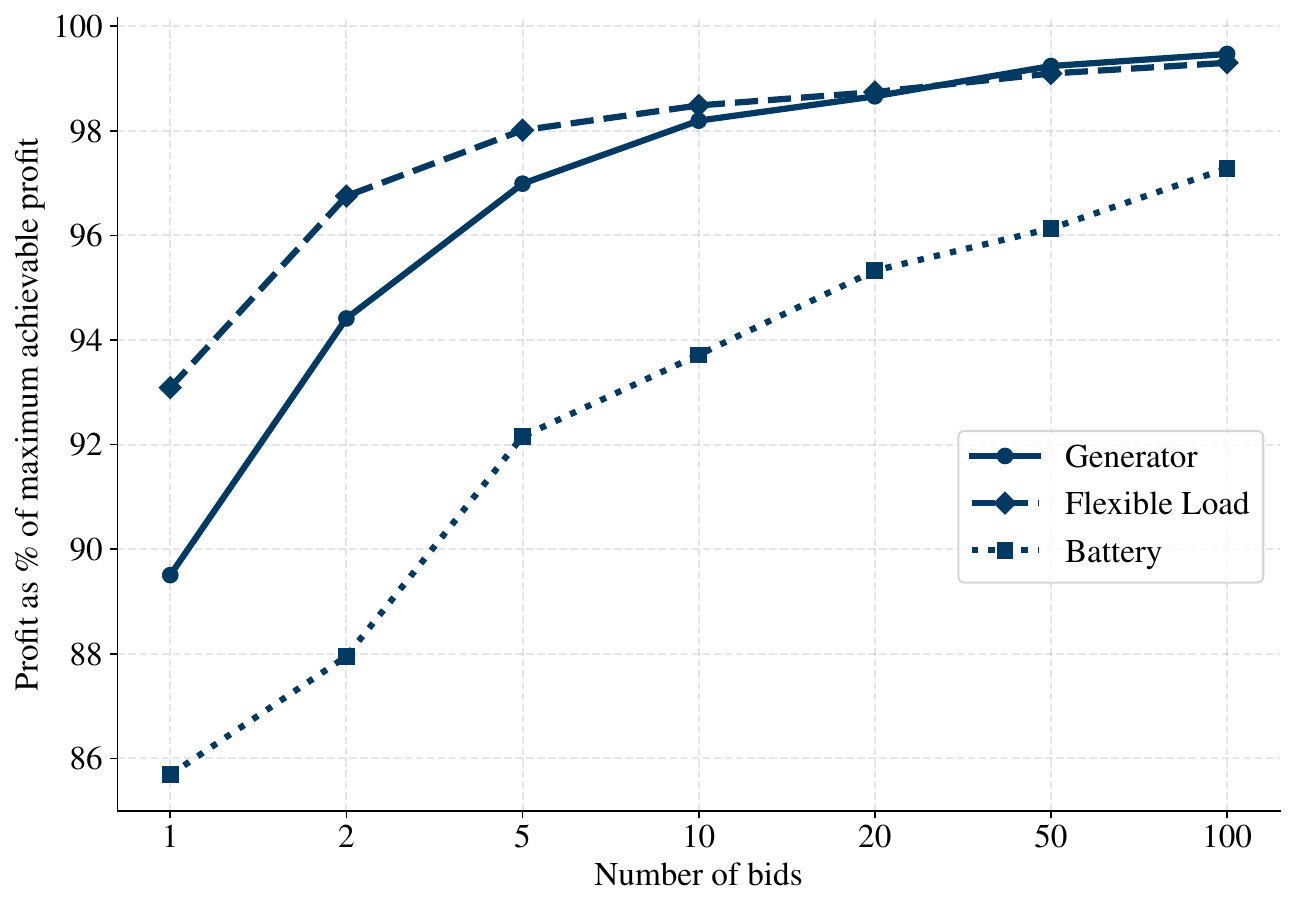}
\label{fig: results bids}
\end{figure}

The results indicate that increasing the number of bids significantly improves profits. Generators and flexible loads, being less sensitive to price uncertainty, perform well with fewer bids. In contrast, battery storage systems, which rely on accurate price forecasts for effective temporal arbitrage, improve with higher bid counts and eventually catch up to the other participants.

\subsection{Number of Scenarios} \label{subsec: results scenarios}

In the second set of experiments, we analyze the loss of profit as a function of the number of scenarios considered when selecting XOR bids by algorithm~\eqref{eq: heuristic II}. Specifically, we generate $B=24$ package bids while accounting for $S$ = 24, 50, 100, 200, and 400 scenarios.

It is worth noting that for $S=24$, algorithm~\eqref{eq: heuristic II} simplifies to algorithm~\eqref{eq: heuristic I} due to $B\ge S$, resulting in package bids that correspond to a solution of the bilevel program~\eqref{eq: stochastic bilevel program}.

To evaluate performance, we sample 100 days and compute the average profit, as detailed in \Cref{subsec: results bids}. The results are summarized in \Cref{fig: results scenarios}. While increasing the number of scenarios from 24 to 50 yields a significant improvement, further increases beyond 50 do not demonstrate a comparable benefit. This limitation can be attributed to our method used to generate scenarios from historical data (\Cref{subsec: price forecasts}); as more scenarios are included, the model incorporates increasingly ``outdated'' information, diminishing its effectiveness.

\begin{figure}[t]
\caption{Average achieved profit in relation to the number of scenarios.}
\centering
\includegraphics[width=0.6\textwidth]{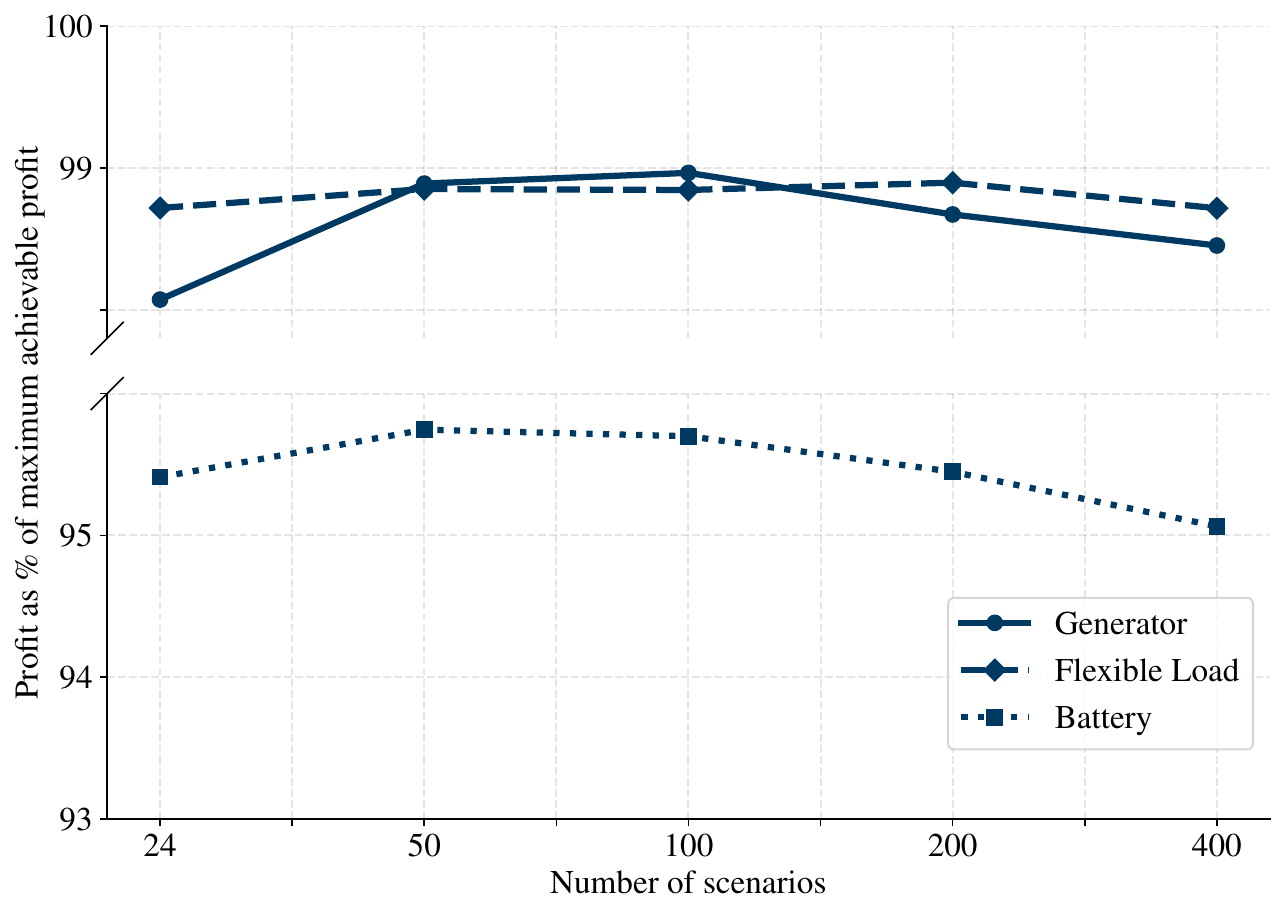}
\label{fig: results scenarios}
\end{figure}

\subsection{Quality of Price Information} \label{subsec: results price information}

In the final set of experiments, we investigate the loss of profit as a function of the quality of the price forecasts. For this analysis, we fix
$B=24$ and generate $S=100$ price scenarios $\lambda_{ds}$ for each day $d$, following the procedure in~\Cref{subsec: price forecasts}. To evaluate the impact of forecast accuracy, we progressively tighten the scenarios towards the actual price $\lambda_d$. Specifically, we define adjusted scenarios as $\lambda_{ds}^\prime = \lambda_{ds} - a \cdot (\lambda_{ds} - \lambda_{d})$, where $a$ is a scaling parameter.

This adjustment simulates improved forecast accuracy, as tightening scenarios closer to the true price can be interpreted as refining the forecasts. We consider five different values of~$a$: $a=0,0.25,0.5,0.75,1$. The case $a=0$ corresponds to the baseline forecast used in previous analyses, while $a=1$ represents perfect information.
To quantify the impact of these adjustments, we computed the Wasserstein distance between the discrete probability distribution formed by the 100 scenarios and the degenerate probability distribution centered on the actual price. This metric is plotted on the x-axis of~\Cref{fig: results information} for each value of $a$.

For the corresponding graph, we sampled 100 days and computed the average profit, as described in~\Cref{subsec: results bids}. The results show that as the accuracy of the forecast improves (i.e., as the Wasserstein distance decreases), profits increase, eventually converging to 100\% in the case of perfect information.

\begin{figure}[t]
\caption{Achieved profit in relation to the quality of price scenarios.}
\centering
\includegraphics[width=0.6\textwidth]{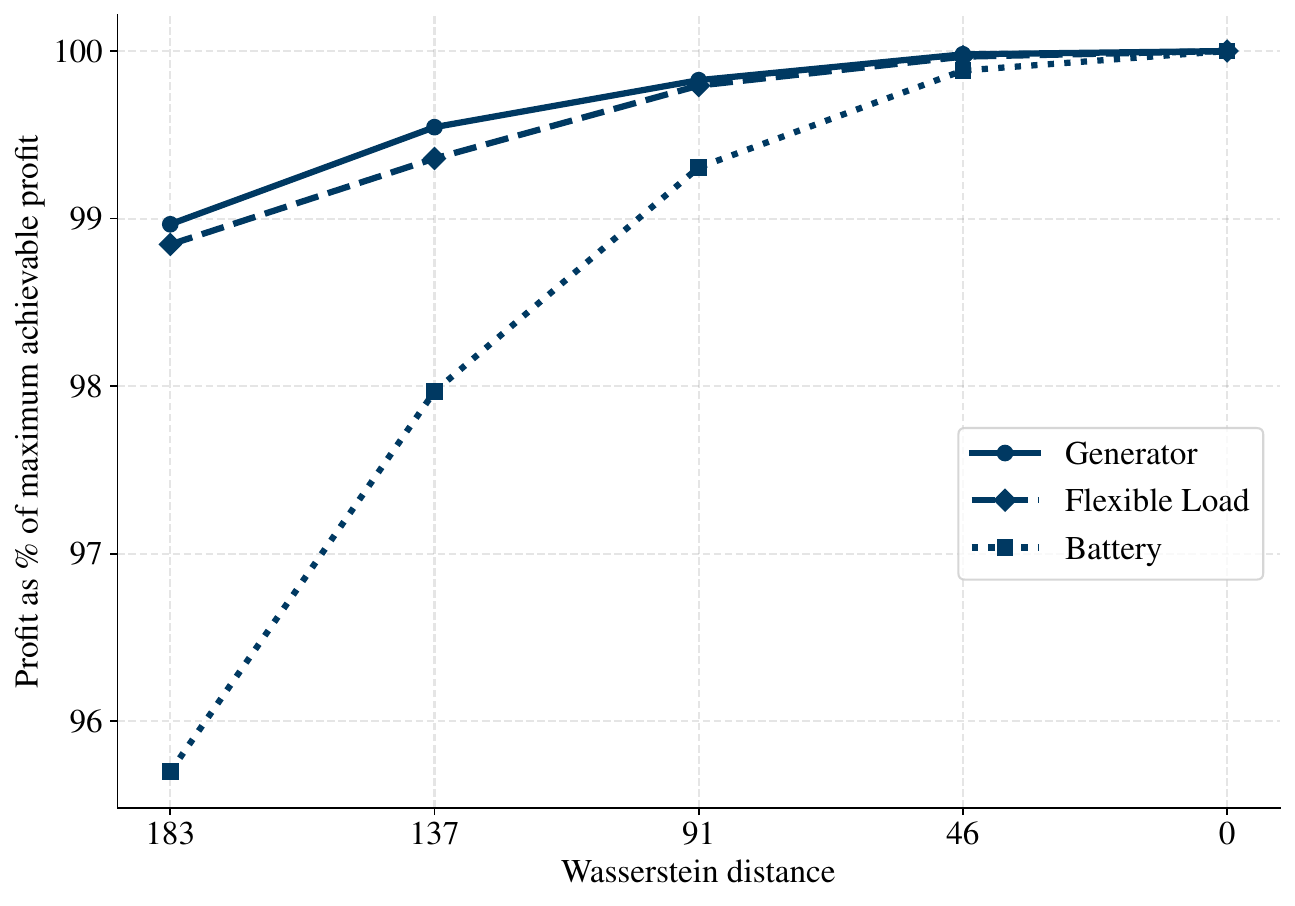}
\label{fig: results information}
\end{figure}

\section{Discussions} \label{sec: disucsions}

If an auctioneer offers both parametric and XOR package bids, agents can choose the format that suits them. Agents whose preferences cannot be well represented by parametric bids or who prefer not to disclose detailed information can use the bid selection procedure in \Cref{sec: bid selection} to minimize profit loss. Agents whose preferences align with parametric bids can use them to completely avoid the risk of profit loss.

\subsection{OR Package Bids} \label{subsec: OR package bids}
An alternative to parametric bids to address the missing bid problem in XOR bids is to use OR package bids~\citep{bichler2023taming}. Unlike XOR bids, which restrict acceptance to a single package, OR bids allow all submitted packages to be accepted independently. By submitting $B$ package bids, an agent effectively communicates interest in all $2^B$ combinations of these packages.
However, OR bids are suitable only for agents with limited patterns of substitution in their preferences~\citep{nisan2006bidding}. This is because OR bids cannot express that two packages are substitutes, where accepting one excludes the other. 
This limitation is especially problematic in electricity auctions, where agents have preferences shaped by intertemporal constraints and upper limits on production or consumption. These constraints create strong substitution patterns: If too many bids are accepted, agents may exceed their capacity to deliver or consume, forcing costly after-auction adjustments.
For agents such as shiftable loads or hydro reservoirs, substitution is even more critical - they operate based on producing or consuming in one period or another, but not both.\endnote{The power exchange EpexSpot also permits ``linking'' multiple package bids, meaning a child bid is accepted only if its parent bid is accepted~\citep{karasavvidis2021optimal}. However, linked package bids face the same limitation as OR bids, as they can only capture a narrow range of substitution patterns in the agent's preferences.}

Ironically, the EpexSpot power exchange allows agents to submit up to 100 OR package bids, but limits the XOR package bids to only 24~\citep{epexspot}. Discussions with practitioners suggest that this restriction historically stems from the so-called flexi orders - a specific type of XOR bid. Flexi orders allowed agents to group several single commodity bids (consumption or production within a single time interval) into an exclusive set. For example, hydro reservoirs used this mechanism to let the auctioneer decide the optimal release time for stored water. Since a day consists of only 24 hours, grouping more than 24 hourly bids was unnecessary, leading to the 24-bid limit. However, when package bids are grouped, rather than single-commodity bids, this limit becomes obsolete. Computational concerns are unlikely the cause, given that up to 100 OR bids can be submitted.

\begin{figure}[t]
\caption{Illustration of XOR bids with a minimum acceptance ratio of 0.}
\centering
\includegraphics[width=0.6\textwidth]{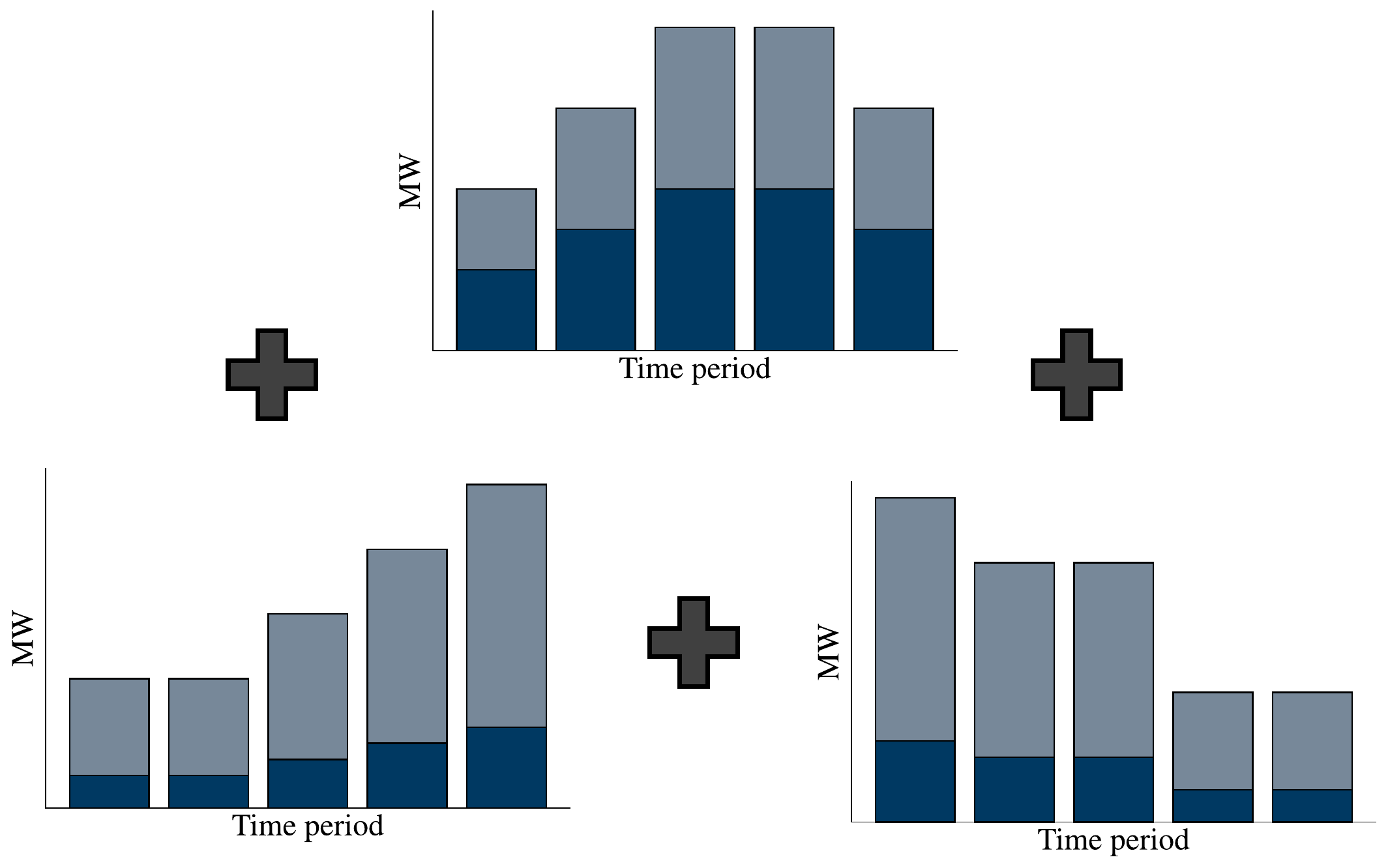}

{\footnotesize \textit{Notes:} The upper profile is accepted by 50\% whereas the lower profiles are accepted by 25\%. In general, any convex combination of those profiles is possible if the minimum acceptance ratio of each profile is 0.}
\label{fig: curtailed XOR bids}
\end{figure}

\subsection{Convex Package Bids} \label{subsec: convex package bids}

Limiting the number of package bids is problematic because it is not causal; package bids only introduce a significant computational burden when they are nonconvex. However, since electricity is divisible, agents can specify a minimum acceptance ratio when submitting package bids. If this ratio is set to 0, the package bid becomes convex, as partial acceptance is allowed without restrictions.
Formally, for each package bid $(x_b,p_b)$, the agent specifies a minimum acceptance ratio $\alpha\in[0,1]$. The actual acceptance rate $\delta_b$ of a bid must then satisfy $\delta_b\in\{0\} \cup [\alpha_b,1]$. If multiple bids $b=1, \ldots,B$ are submitted as XOR bids, the constraint $\sum_{b\in\mathcal{B}} \delta_b \le 1$ must be satisfied. Finally, the agent is allocated a package $x^\prime = \sum_{b\in\mathcal{B}} \delta_b \cdot x_b$, which may be a convex combination of its package bids~(\Cref{fig: curtailed XOR bids}).

Looking at the relative use of minimum acceptance ratios for package bids in the German bidding zone in 2023, we observe that 16.7\% of the bids had a ratio of 0 (that is, convex bids), 2.7\% were submitted with a ratio strictly between 0 and 1, and 80.6\% of the bids had a ratio of 1 - a so-called ``fill-or-kill'' condition.$^9$ While there are certainly agents with nonconvex preferences who need to submit package bids with $\alpha>0$, this clearly does not apply to every agent.
Auctioneers could mitigate the welfare loss due to missing bids while still keeping optimization tractable by imposing a tight limit on package bids with $\alpha>0$, but allowing more package bids with $\alpha=0$.

\subsection{Aggregators and Portfolios} \label{subsec: aggregators and portfolios}

Potential synergies between assets that consume or produce power encourage pooling them together. A notable example is the growing role of demand-side aggregators for distributed energy resources~\citep{burger2017review}. To fully harness these synergies for societal welfare, it is necessary to allow coordinated bidding by such entities. XOR bids provide a practical format for this purpose, as designing parametric bids for aggregated resources is not a straightforward task~\citep{liu2015extending,herrero2020evolving}.

However, while synergies may justify bidding as a portfolio, concerns about potential market power abuse also arise. Ideally, large participants, such as major thermal power plants, would employ parametric bid formats to transparently disclose detailed cost information to the auctioneer, facilitating the detection of market power abuse~\citep{adelowo2024redesigning}.
When participants avoid such transparency, for example, through self-scheduling or by opting for package or single-commodity bids, they impede the auctioneer’s ability to detect and mitigate market manipulation. Permitting large participants to cite ``synergies'' among assets as a rationale to avoid parametric bid formats thus risks reducing welfare~\citep{ahlqvist2022survey}.

\section{Conclusions} \label{sec: conclusion}

Effective bid formats that allow agents to express preferences accurately are vital for efficient auctions. Auctioneers benefit from offering technology-specific parametric bid formats alongside universal and simple XOR package bids. To minimize welfare losses from limits on package bids, it helps to share as much information as possible about likely uniform prices. Allowing many convex package bids while placing tight limits only on nonconvex ones can keep optimization manageable while mitigating welfare losses.

Participants are encouraged to use parametric bid formats whenever possible. If these are not an option, it is advisable to submit the maximum number of packages allowed. When the number of desired bids exceeds the limit, bid selection algorithms can assist in prioritizing those with the highest expected profit.

This study assumes a competitive equilibrium and does not address price manipulation through market power~\citep{reguant2014complementary} or exploitation of imperfect equilibrium outcomes~\citep{byers2022auction}. Theoretical and empirical evidence by~\cite{starr1969quasi},~\cite{bichler2018matter} or~\cite{hubner2025approximate} suggest that while uniform prices may not perfectly align with competitive equilibrium outcomes, they are often also not too far off.
Further research could help to better understand the welfare losses associated with imperfect bid formats when competitive equilibrium is not assumed.

\section*{Acknowledgements}{The research published in this report was carried out with the support of the Swiss Federal Office of Energy (SFOE) as part of the
SWEET PATHFNDR project. The authors bear sole responsibility for the conclusions and the results. The authors would like to thank Prof. Quentin Lété and Prof. Daniel Kuhn for their helpful comments and insightful discussions.}

\begingroup \parindent 0pt \parskip 0.0ex \def\enotesize{\small} \theendnotes \endgroup

\appendix 
\setcounter{section}{0}
\renewcommand{\thesection}{Appendix \Alph{section}}

\section{Proof of Theorem 1}\label{app: proof theorem wasserstein}

Denote the function in~\eqref{eq: max profit} by $\psi(\lambda)$ and the function in~\eqref{eq: profit under restriction} by $\varphi(\lambda)$. The profit loss is then given by $\Gamma(\lambda) = \psi(\lambda) - \varphi(\lambda)$. In a first step, we show that both $\psi$ and $\varphi$ are Lipschitz-continuous with constant $L=\max \big\{ \|x\|_2 \ \big | \ x\in\mathbb{R}^T, v(x) \neq -\infty \big\}$.
In microeconomics, $\psi$ is known as the indirect utility function, which is known to be continuous if $v$ is continuous~\citep{mas1995microeconomic}. However, we cannot necessarily assume $v$ to be continuous in electricity markets. Despite this, we can still establish the stronger property of Lipschitz continuity for $\psi(\lambda)$, as follows:

\begin{lemma}\label{lemma: lipschitz continuity 1}
For all $\lambda_1, \lambda_2 \in \Omega$ holds $| \psi(\lambda_1) - \psi(\lambda_2) | \ \le \ L \cdot \| \lambda_1 - \lambda_2 \|_2$.
\end{lemma}
\begin{proof}
Let $\lambda_1,\lambda_2\in\Omega$. Distinguish two cases: (i) $\psi(\lambda_1) \ge \psi(\lambda_2)$ and (ii) $\psi(\lambda_1) < \psi(\lambda_2)$. 

\textit{Case (i):} Let $x_1\in\argmax v(x) - \langle \lambda_1, x \rangle$. From $\psi(\lambda_1) \ge \psi(\lambda_2)$ and $\psi(\lambda_2)\ge v(x_1) - \langle\lambda_2,x_1\rangle$ follows
\begin{subequations} \label{eq: lipschitz proof I}
\begin{align*} 
&| \psi(\lambda_1) - \psi(\lambda_2) |   \\
&= \psi(\lambda_1) - \psi(\lambda_2) \\
&\le v(x_1) - \langle \lambda_1, x_1 \rangle - v(x_1) + \langle \lambda_2, x_1 \rangle \\
&= \langle \lambda_2-\lambda_1, x_1 \rangle .
\end{align*}
\end{subequations}
Combining this with the \textit{Cauchy-Schwarz inequality}, we arrive at
\begin{align*}
| \psi(\lambda_1) - \psi(\lambda_2) | 
& \le \ | \langle \lambda_2-\lambda_1, x_1 \rangle |  \\
& \le \  \|x_1 \|_2 \cdot \| \lambda_1 - \lambda_2 \|_2 .
\end{align*}

\textit{Case (ii):} Let $x_2\in\argmax v(x) - \langle \lambda_2, x \rangle$. By similar arguments as in case (i) follows 
$$| \psi(\lambda_1) - \psi(\lambda_2) | \le \|x_2 \|_2 \cdot \| \lambda_1 - \lambda_2 \|_2 .$$

Finally, the statement follows by $\forall \lambda\in\Omega$ and $\forall x \in\argmax v(x) - \langle \lambda, x \rangle$ holds $\|x\|_2 \le \max \big\{ \|x^\prime\|_2 \ \big | \ x^\prime\in\mathbb{R}^T, v(x^\prime) \neq -\infty \big\}$. 
\end{proof}

\begin{lemma} \label{lemma: lipschitz continuity 2}
For all $\lambda_1, \lambda_2 \in \Omega$ holds  $| \varphi(\lambda_1) - \varphi(\lambda_2) | \ \le \ L \cdot \| \lambda_1 - \lambda_2 \|_2$.
\end{lemma}
\begin{proof}
Follows by \Cref{lemma: lipschitz continuity 1} and $\{\bar{x}_1, \ldots, \bar{x}_S\}$ being a subset of every possible package. 
\end{proof}

Lipschitz-continuity of $\psi$ and $\varphi$ is important as the Wasserstein distance is of so-called $\zeta$-structure which is closely linked to Lipschitz-continuity~\citep{rachev2013}.
With the help of \Cref{lemma: lipschitz continuity 1} and \Cref{lemma: lipschitz continuity 2}, the proof of \Cref{theorem: wasserstein} can be carried out as follows:

\textnormal{\textbf{Proof of \Cref{theorem: wasserstein}}}
First of all, we relate $\text{E} \big[ \Gamma(\lambda) \big]$ to distributions $\mathcal{P}$ and $\mathcal{Q}$ by:
\begin{subequations} \label{eq: proof theorem}
\begin{align}
& \text{E}\big[ \Gamma(\lambda) \big] \\
& = \int_\Omega \Gamma(\lambda) \; \mathcal{P}(d\lambda)  \\ 
& = \int_\Omega \Gamma(\lambda) \; \big(\mathcal{P}+\mathcal{Q}-\mathcal{Q}\big)(d\lambda) 
\end{align}
Note that $\int_\Omega \Gamma(\lambda) \; \mathcal{Q}(d\lambda)=0$ since the bidding strategy~\eqref{eq: heuristic I} ensures that if a scenario $\lambda_s$ of the distribution $\mathcal{Q}$ occurs, then the agent is maximising its profit. Hence, we can continue by:
\begin{align}
& = \int_\Omega \Gamma(\lambda) (\mathcal{P}-\mathcal{Q})(d\lambda)  \\
& = \int_{\Omega} \Gamma(\lambda_1) \; \mathcal{P}(d\lambda_1) - \int_{\Omega} \Gamma(\lambda_2) \; \mathcal{Q}(d\lambda_2) . \label{eq: proof theorem II} 
\end{align}
By definition of $\mathcal{P}$ and $\mathcal{Q}$ as probability measures follows $\int_\Omega \mathcal{P}(d\lambda_1) = \int_\Omega \mathcal{Q}(d\lambda_2)=1$ and hence, we can continue \eqref{eq: proof theorem II} by:
\begin{align}
& = \int_{\Omega \times \Omega} \Gamma(\lambda_1) \; \mathcal{P}(d\lambda_1)\mathcal{Q}(d\lambda_2) \nonumber \\
& - \int_{\Omega \times \Omega} \Gamma(\lambda_2) \; \mathcal{P}(d\lambda_1)\mathcal{Q}(d\lambda_2) \label{eq: proof theorem III} . 
\end{align}
Since the Euclidean space is a \textit{Polish space}, Theorem 4.1 from \cite{villani2009} ensures the existence of a solution $\mu^\ast(\lambda_1,\lambda_2)$ to \eqref{eq: wasserstein distance}. By definition, $\mu^\ast$ is a joint probability distribution with marginals $\mathcal{P}$ and $\mathcal{Q}$. Consequently, we can continue \eqref{eq: proof theorem III} by:
\begin{align}
& = \int_{\Omega\times\Omega} \Gamma(\lambda_1) - \Gamma(\lambda_2) \;\; \mu^\ast(d\lambda_1,d\lambda_2) \label{eq: proof theorem IV} . 
\end{align}
Using the definition of the \textit{Lebesgue integral}, this can be bounded above by:
\begin{align}
& \le \int_{\Omega\times\Omega} | \Gamma(\lambda_1) - \Gamma(\lambda_2) | \;\; \mu^\ast(d\lambda_1,d\lambda_2) \label{eq: proof theorem V} . 
\end{align}
By \Cref{lemma: lipschitz continuity 1} and \Cref{lemma: lipschitz continuity 2} follows that $\Gamma(\lambda) = \psi(\lambda) - \varphi(\lambda)$ is Lipschitz continuous with constant $L = 2 \cdot \max \big\{ \|x\|_2 \ \big | \ x\in\mathbb{R}^T, v(x) \neq -\infty \big\}$. Hence, we can continue \eqref{eq: proof theorem V} by:
\begin{align}
& \le L \cdot \int_{\Omega\times\Omega} \| \lambda_1, \lambda_2 \|_2 \;\; \mu^\ast(d\lambda_1,d\lambda_2) .
\end{align}
\end{subequations}
Finally, by $\mu^\ast$ being a solution to \eqref{eq: wasserstein distance} follows the statement.

\section{Package Bid Selection as Bilevel Knapsack Problem}\label{app: bilevel knapsack}

The package bid selection problem~\eqref{eq: stochastic bilevel program III}, along with its reformulations~\eqref{eq: binary program} and~\eqref{eq: binary program I}, can be equivalently expressed as the following bilevel knapsack problem:  
\begin{subequations}\label{eq: knapsack}
\begin{align} 
\max_{\delta} & \quad \sum_{s\in\mathcal{S}} \pi_s \cdot \sum_{k\in\mathcal{K}} U_{ks} \cdot \gamma_{ks} \label{eq: knapsack 0} \\
\text{s.t.} & \quad \sum_{k\in\mathcal{K}} \delta_k \le B \label{eq: knapsack 1} \\
& \quad \delta_k \in \{0,1\} \quad \forall k \in \mathcal{K} \label{eq: knapsack 2}  \\
& \quad \big(\gamma_{ks}\big)_{k\in\mathcal{K},s\in\mathcal{S}} \in \label{eq: knapsack 3} \\
& \quad \argmax_{\gamma} \quad \sum_{s\in\mathcal{S}} \pi_s \cdot \sum_{k\in\mathcal{K}} U_{ks} \cdot \gamma_{ks}\label{eq: knapsack 4}  \\
& \qquad \quad \text{s.t.} \quad \gamma_{ks} \le \delta_k \quad \forall s\in\mathcal{S}, \; k\in\mathcal{K} \label{eq: knapsack 5} \\
& \qquad \qquad \quad \sum_{k\in\mathcal{K}} \gamma_{ks} \le 1 \quad \forall s\in\mathcal{S} \label{eq: knapsack 6} \\
& \qquad \qquad \quad \gamma_{ks} \in \{0,1\} \quad \forall s\in\mathcal{S}, \; k\in\mathcal{K} . \label{eq: knapsack 7}
\end{align}
\end{subequations}
where $U_{ks} = \bar{v}_k - \langle \lambda_s, \bar{x}_{k} \rangle$ represents the utility of selecting package $k$ in scenario $s$.
The leader \eqref{eq: knapsack 0}-\eqref{eq: knapsack 3} (``bidder'') selects up to $B$ packages to include in the knapsack.
The follower \eqref{eq: knapsack 4}-\eqref{eq: knapsack 7} (``auctioneer'') can then take one package out of the knapsack for each scenario $s\in\mathcal{S}$.
The leader and follower have the same objective; they both try to maximize the expected profit of the selected packages.

This bilevel knapsack problem is especially tractable since the leader and follower share the same objective, causing the bilevel structure to collapse into a single-level problem with a totally unimodular constraint matrix (\Cref{proposition: lp relaxation}).

\section{Optimizing the Conditional Value at Risk} \label{app: conditional value at risk}

In formulating the bid selection problem~\eqref{eq: stochastic bilevel program}, we assumed a risk-neutral agent maximizing expected profit. However, risk-averse agents may prefer to optimize expected profit in the worst $1-\beta\%$ of cases - that is, to maximize the $\beta$-CVaR. Using the well-known reformulation by \cite{rockafellar2000optimization}, this can be expressed as:
\begin{align*}
    \max_{x_b,p_b, \alpha, \eta_s} & \quad \alpha - \frac{1}{1-\beta} \cdot \sum_{s\in\mathcal{S}} \pi_s \cdot \eta_s \\
    \text{s.t.} & \quad \eta_s \ge \alpha - v(x_s^\ast) + \langle \lambda_s, x_s^\ast \rangle  \quad \forall s \in \mathcal{S} \\
    & \quad \eta_s \ge 0, \; \alpha \in \mathbb{R} \quad \forall s \in \mathcal{S} \\
    & \quad \eqref{eq: most profitable bid} \;\;\; \text{and} \;\; (x_b,p_b)\in\mathbb{R}^{T+1} .
\end{align*}
Theorem 1 in \cite{rockafellar2000optimization} shows that the objective function $\alpha - \frac{1}{1-\beta} \sum_{s\in\mathcal{S}} \pi_s \cdot \eta_s$ corresponds to the $\beta$-CVaR, with $\alpha$ representing the $\beta$-VaR at the optimum. Theorem 2 confirms the correctness of the reformulation.

\section{Thermal Generator} \label{app: thermal generator}

The thermal generator model used here is based on the classic unit commitment approach proposed by~\cite{carrion2006computationally}, which was also applied by~\cite{karasavvidis2021optimal} in the context of day-ahead auction participation of generators. The cost function $v(x)$ is defined by \eqref{eq: thermal generator 14}, while the operational constraints governing the model are expressed in \eqref{eq: thermal generator 1}-\eqref{eq: thermal generator 13}. For a detailed description of the model, see~\cite{carrion2006computationally} or~\cite{karasavvidis2021optimal}.

\begin{subequations} \label{eq: thermal generator}
The power generated and sold by the thermal generator is the sum of the power output $g_{tk}$ from each block $k\in\mathcal{K}$ of the plant:
\begin{align}
x_{t} = \sum_{k \in \mathcal{K}} g_{tk} \qquad \forall t \in \mathcal{T} . \label{eq: thermal generator 1}
\end{align}
The production limit $P^{Max}_k$ for each block $k\in\mathcal{K}$ is enforced by:
\begin{align}
- g_{tk} \le u_{t} \cdot P^{Max}_k \qquad \forall k \in \mathcal{K}, \; \forall t \in \mathcal{T} , \label{eq: thermal generator 2}
\end{align}
where $u_{t}$ is the binary variable describing whether the generator is committed or not. Remember that $x<0$ means selling power, leading to the minus in front of $g_{tk}$. The minimum stable generation $P^{Min}$ is enforced by:
\begin{align}
P^{Min} \cdot u_{t} \le \sum_{k\in\mathcal{K}} - g_{tk} \qquad \forall t \in \mathcal{T} . \label{eq: thermal generator 3} 
\end{align}
The ramp-up limit $R^{Up}$ and the ramp-down limit $R^{Dn}$ are enforced by:
\begin{align}
- R^{Dn} \le x_{t-1} - x_{t} \le R^{Up} \qquad \forall t \in \mathcal{T} . \label{eq: thermal generator 4} 
\end{align}
Start-up costs $c^{up}_{t}$ and shut-down costs $c^{up}_{t}$ in scenario $s$ and period $t$ can be expressed by constraints:
\begin{align}
c^{up}_{t} \ge (u_{t} - u_{t-1}) \cdot K^{Up} \qquad \forall t \in \mathcal{T} \label{eq: thermal generator 5} \\
c^{dn}_{t} \ge (u_{t-1} - u_{t}) \cdot K^{Dn} \qquad \forall t \in \mathcal{T} , \label{eq: thermal generator 6} 
\end{align}
where $K^{Up}$ represent the start-up and $K^{Dn}$ the shut-down cost. The initial-on-time $H^{Up}$ respectively initial-off-time $H^{Dn}$ is enforced by constraints: 
\begin{align}
\sum_{t=1}^{H^{Up}} (1- u_{t}) = 0   \label{eq: thermal generator 7} \\
\sum_{t=1}^{H^{Dn}} u_{t} = 0  , \label{eq: thermal generator 8} 
\end{align}
Finally, the minimum-up-time $T^{Up}$ and minimum-down time $T^{Dn}$ is enforced by the constraints:
\begin{align}
T^{Up} \cdot (u_{t} - u_{t-1}) \le \sum_{j=t}^{t+T^{Up}-1} u_{t} \nonumber \\
\qquad \forall t \in [H^{Up}+1, \ldots, T - T^{Up}+1] \label{eq: thermal generator 9} 
\end{align}
\begin{align}
-T^{Dn} \cdot (u_{t} - u_{t-1}) \le \sum_{j=t}^{t+T^{Dn}-1} (1-u_{t}) \nonumber \\
\qquad \forall t \in [H^{Dn}+1, \ldots, T - T^{Dn}+1] \label{eq: thermal generator 10} 
\end{align}
\begin{align}
\sum_{j=t}^{T} u_{j} - (u_{t} - u_{t-1}) \ge 0 \nonumber \\
\qquad \forall  \; t \in [T-T^{Up}+2,\ldots,T] \label{eq: thermal generator 11} 
\end{align}
\begin{align}
\sum_{j=t}^{T} (1 - u_{sj}) + (u_{t} - u_{t-1}) \ge 0 \nonumber \\
\qquad \forall  \; t \in [T-T^{Dn}+2,\ldots,T] , \label{eq: thermal generator 12} 
\end{align}
designed by \cite{carrion2006computationally} to obtain a computationally efficient formulation. Finally, the constraints on the auxiliary variables are given by: 
\begin{align}
g_{tk} \le 0, \; u_{t}\in\{0,1\}, \; c^{up}_{t}, c^{dn}_{t} \ge 0 \nonumber \\
\qquad \forall k \in \mathcal{K}, \;  \; t \in \mathcal{T} , \label{eq: thermal generator 13} 
\end{align}
and the cost of power generation by:
\begin{align}
v(x) = \sum_{t \in\mathcal{T}} - K^{F} \cdot u_{t} - c^{up}_{t} - c^{dn}_{t}  + \sum_{k \in \mathcal{K}} K^{M}_{k} \cdot g_{tk} , \label{eq: thermal generator 14} 
\end{align}
where $K^{F}$ are the no-load costs and $K^{M}_{k}$ are the marginal cost of production in block $k\in\mathcal{K}$.
\end{subequations}

\section{Battery Storage System} \label{app: battery}

The following model for a battery storage system is commonly used for power systems applications of a battery, for instance, in \cite{de2019implications} or \cite{karasavvidis2023optimal}.

\begin{subequations} \label{eq: battery}
The buy and sell volume in hour $t$ and scenario $s$ is the difference between charging $g_{t}$ and discharging $d_{t}$:
\begin{align}
x_{t} = g_{t} - d_{t} \qquad \forall t \in \mathcal{T} \label{eq: battery 1} 
\end{align}
Charging and discharging are bounded by the charging respectively discharging limits $\overline{g}$ and $\overline{d}$:
\begin{align}
0 \le g_{t} \le \overline{g} \cdot \delta_{t} \qquad \forall t \in \mathcal{T} \label{eq: battery 2} \\
0 \le d_{t} \le \overline{d} \cdot (1-\delta_{t}) \qquad \forall t \in \mathcal{T} \label{eq: battery 3} \\
\end{align}
The binary variable $\delta_{t} \in \{0,1\}$ indicates whether the storage is charging or discharging. The state-of-charge $e_{t}$ is given by the balance equations:
\begin{align}
e_{t} = e_{t-1} + n^g \cdot g_{t} - \frac{d_{t}}{n^d} \qquad \forall t \in \mathcal{T} , \label{eq: battery 5} 
\end{align}
where $n^g$ is the charging, and $n^d$ is the discharging efficiency.
The state-of-charge must be kept between its lower $\underbar{E}$ and upper $\overline{E}$ state-of-charge limit:
\begin{align}
-\underbar{E} \le e_{t} \le \overline{E} \qquad \forall t \in \mathcal{T} \label{eq: battery 6} 
\end{align}
We assume that the state-of-charge must equal the initial $E_0$:
\begin{align}
e_{T} = E_0 \label{eq: battery 7} 
\end{align}
Finally, we ignore degradation and maintenance costs, and the marginal costs are thus zero:
\begin{align}
v(x) = 0 \label{eq: battery 8} 
\end{align}
\end{subequations}

\section{District Heating Utility} \label{app: district heating utility}

The model of the district heating utility is oriented on the model in~\cite{bobo2021price}.
The task of the utility is to serve the heat load $D_t$ in every hour $t$. For that, it can produce heat by a heat pump consuming power $x_{t}$, produce heat by a gas boiler consuming gas $y_{t}$, or curtail the load $z_{t}$. Moreover, it can discharge heat storage by $d_{t}$ or charge heat storage by $g_{t}$.

\begin{subequations} \label{eq: demand response}
The heat balance equations are given by:
\begin{align}
\eta^{Hp} \cdot x_{t} + \eta^{Gas} \cdot y_{t} + d_{t} - g_{t} = D_t - z_{t} \qquad \forall t \in \mathcal{T} , \label{eq: demand response 1} \end{align}
where $\eta^{Hp}$ is the efficiency of the heat pump and $\eta^{Gas}$ is the efficiency of the gas boiler. The state-of-charge of the heat storage is given by:
\begin{align}
e_{t} = (1-\eta^{Loss}) \cdot e_{t-1} + g_{t} - d_{t} \qquad \forall t \in \mathcal{T} \label{eq: demand response 2} , \\
\end{align}
where $\eta^{Loss}$ is the percentage of heat energy lost between two hours. We assume that the state-of-charge must be, in the end, equal to the initial state-of-charge $E_0$: 
\begin{align}
    e_{Ts} = E_0 \label{eq: demand response 3} 
\end{align}
The state-of-charge, charge, and discharge, as well as the heat generated by the gas boiler and heat pump, are bounded
\begin{align}
e_{t} \in [0,\overline{e}], \;g_{t} \in [0,\overline{g}], \; d_{t} \in [0,\overline{d}], \; z_{t}\in [0,D_t], \nonumber \\
\qquad x_{t} \in [0,\overline{x}], \; y_{t} \in [0,\overline{y}] \qquad \forall  \; t \in \mathcal{T} \label{eq: demand response 4}   
\end{align}
Finally, the valuation comprises the revenue from serving the load and the cost of gas:
\begin{align}
v(x) = K \cdot (D_t - z_{t}) - C \cdot y_{t} \label{eq: demand response 5} ,
\end{align}
 where $K$ is the value of the served load and $C$ is the cost of gas.
\end{subequations}

\bibliography{refs} 

\end{document}